\documentclass[12pt]{article}
\usepackage{graphicx} 
\RequirePackage{amsthm,amsmath,amsfonts,amssymb}
\RequirePackage[numbers]{natbib}
\RequirePackage[colorlinks,citecolor=blue,urlcolor=blue]{hyperref}
\usepackage[utf8]{inputenc}
\usepackage{mathtools} 
\usepackage{booktabs} 
\usepackage{tikz, pgfplots} 
\usepackage[a4paper,width=150mm,top=25mm,bottom=25mm]{geometry}
\usepackage{amsmath}
\usepackage{amssymb}
\usepackage{mathtools}
\usepackage{xcolor}
\usepackage{caption}
\usepackage{subcaption}
\usepackage{verbatim}
\usepackage{titlesec}
\usepackage{csquotes}
\usepackage{authblk}

\newtheorem{theorem}{Theorem}
\newtheorem{proposition}{Proposition}
\newtheorem{lemma}{Lemma}
\newtheorem{corollary}{Corollary}

\theoremstyle{definition}

\newtheorem{example}{Example}
\newtheorem{remark}{Remark}
\newtheorem{assumption}{Assumption}

\newcommand{\X}{\mathbb{X}}

\newcommand{\Z}{\mathbb{Z}}
\newcommand{\R}{\mathbb{R}}
\newcommand{\E}{\mathbb{E}}
\newcommand{\Prb}{\mathbb{P}}
\newcommand{\C}{\mathbb{C}}
\newcommand{\F}{\mathcal{F}}
\newcommand{\G}{\mathcal{G}}
\newcommand{\D}{\mathcal{D}}

\newcommand{\SNS}{\mathrm{S}}
\newcommand{\ENSO}{\mathrm{ENSO}}
\newcommand{\AMO}{\mathrm{AMO}}
\newcommand{\NAO}{\mathrm{NAO}}
\newcommand{\AOD}{\mathrm{AOD}}

\newcommand{\lag}{\mathbf{l}}
\newcommand{\Lag}{\mathbf{L}}
\newcommand{\Fsep}{H}
\newcommand{\fsep}{h}
\newcommand{\csep}{g}
\newcommand{\intsep}{f}
\newcommand{\rfsep}{\mathrm{re}}
\newcommand{\ifsep}{\mathrm{im}}

\newcommand{\sd}{\mathcal{S}}

\newcommand{\re}{\mathrm{Re}}
\newcommand{\imag}{\mathrm{{Im}}}
\newcommand{\acov}{\mathrm{acov}}
\newcommand{\pa}{\mathrm{pa}}
\newcommand{\paths}{\mathrm{P}}
\newcommand{\internal}{\mathrm{I}}

\newcommand{\anc}{\mathrm{anc}}

\titleformat{\title}{display}{}{}{}[]
\titleformat{\section}[runin]{\normalsize\bfseries}{\thesection.}{0.5em}{}[.]
\titleformat{\subsection}[runin]{\normalsize\itshape}{\thesubsection.}{0.3em}{}[.]

\providecommand{\keywords}[1]
{
  \small	
  \textbf{\textit{Keywords---}} #1
}

\title{\large\textbf{ASYMPTOTIC UNCERTAINTY IN THE ESTIMATION OF FREQUENCY DOMAIN CAUSAL EFFECTS FOR LINEAR PROCESSES}}
\author[1]{\normalsize Nicolas-Domenic Reiter \thanks{nicolas-domenic.reiter@dlr.de}}
\author[2,1]{ \normalsize Jonas Wahl \thanks{wahl@tu-berlin.de}}
\author[3]{\normalsize Gabriele C. Hegerl \thanks{Gabi.Hegerl@ed.ac.uk}}
\author[1,2,4]{\normalsize Jakob Runge \thanks{jakob.runge@tu-dresden.de}}
\affil[1]{\normalsize German Aerospace Center (DLR), Institute of Data Science, Jena, Germany}
\affil[2]{\normalsize Technische Universität Berlin, Berlin, Germany}
\affil[3]{The University of Edinburgh, School of Geoscience, Edinburgh, United Kingdom}
\affil[4]{\normalsize Center for Scalable Data Analytics and Artificial Intelligence (ScaDS.AI) Dresden/Leipzig, TU Dresden, Germany}
\date{}
\begin{document}

\maketitle

\begin{abstract}
Structural vector autoregressive (SVAR) processes are commonly used time series models to identify and quantify causal interactions between dynamically interacting processes from observational data. The causal relationships between these processes can be effectively represented by a finite directed process graph - a graph that connects two processes whenever there is a direct delayed or simultaneous effect between them. Recent research has introduced a framework for quantifying frequency domain causal effects along paths on the process graph. This framework allows to identify how the spectral density of one process is contributing to the spectral density of another. In the current work, we characterise the asymptotic distribution of causal effect and spectral contribution estimators in terms of algebraic relations dictated by the process graph. Based on the asymptotic distribution we construct approximate confidence intervals and Wald type hypothesis tests for the estimated effects and spectral contributions. Under the assumption of causal sufficiency, we consider the class of differentiable estimators for frequency domain causal quantities, and within this class we identify the asymptotically optimal estimator. We illustrate the frequency domain Wald tests and uncertainty approximation on synthetic data, and apply them to analyse the impact of the 10 to 11 year solar cycle on the North Atlantic Oscillation (NAO). Our results confirm a significant effect of the solar cycle on the NAO at the 10 to 11 year time scale.
\end{abstract}


\section{Introduction}\label{sec:intro}
Many scientific questions are causal questions; it is often the objective to test and quantify a hypothesised causal effect between two or more processes. Causal inference \citep{pearl2009causality, peters2017elements} provides a formal language for reasoning about fundamental problems in questions of cause and effect. The underlying mathematical object of any causal problem is a causal model, which consists of a causal graph and a model compatible with the graph.
In causal modelling it is often important to consider the temporal order of causal relations \citep{runge2023causal}. The causal structure of a system that is modelled to evolve over a discrete set of time points is often represented by an infinite graph, where a vertex represents the state of a modelled process at a given time, and edges identify direct delayed or simultaneous causal relationships between processes.
In the literature, graphs of this type are known as full time graphs, time series DAG or chain graphs \citep{peters2017elements, gerhardus2021characterization, dahlhaus2003causality, gerhardus2023projecting}. The infinite time series graph can be reduced to a finite graph whose vertices correspond to the modelled processes. An edge on this finite graph signals the existence of at least one lagged or contemporaneous effect between the respective processes. This reduced graph is sometimes referred to as the summary graph or process graph \citep{peters2017elements, reiter2023formalising}.

Among the possible models for time series graphs are the classic and widely used SVAR processes \citep{lutkepohl2005new, brockwell2009time}. In previous work \cite{reiter2023formalising}, we showed that SVAR processes can be reformulated as a linear structural causal model of stochastic processes. This type of causal model we called \emph{structural equation process (SEP)} \citep{reiter2023formalising}. In this formulation, each edge on the process graph is assigned a linear filter and each vertex is assigned an independent noise process, so that each SVAR component process can be written as the noise process plus the sum of the filtered parent processes.  
An implication of this reformulation is a notion of process-level causal effects and a generalisation of the classical path rule \citep{wright1934method, Wright1921CorrelationAndCausation}. It turns out that the computations of these process-level causal effects become more tractable in the frequency domain after the application of the Fourier transform. The Fourier transformed causal effect filters are rational functions on the complex unit circle parameterised by the underlying SVAR coefficient vector \cite{reiter2023formalising}. These functions describe direct and indirect impulse response relationships in the frequency domain between the SVAR component processes. When one process directly or indirectly causes another, the spectral density of the latter is partially determined by the spectral density of the causing process. This part of the spectral density can be quantified and it is called the spectral contribution of the causing process.

Estimating causal effects and spectral contributions in the frequency domain is of potential interest in many scientific applications \citep{runge2019inferring, Seth3293}. However, in order to draw conclusions from such estimates, it is necessary to have means of assessing the associated statistical uncertainty. One possibility to approach the uncertainty of estimators is to consider their asymptotic behaviour, i.e. the behaviour of the estimator for large sample sizes. In this paper we study aspects of first-order asymptotic theory for estimators of causal quantities in the frequency domain. The present paper is structured as follows. 

In Section \ref{sec: preliminaries}, we recall the necessary definitions and notions from causal time series modelling. In Section \ref{sec: asymptotic distribution}, we establish rational expressions for the asymptotic covariance of ordinary least squares (OLS) based estimators for causal effects and spectral contributions in the frequency domain. Based on these expressions we formulate Wald type tests, and we approximate the uncertainty of the effect and contribution estimators. 
For our computations we assume that there are no unobserved processes and that the structure of the process graph and the contemporaneous links on the time series graph are known. However, we do not need to know the exact time lags with which one process possibly drives another. To illustrate the uncertainty approximation in OLS-based estimates of causal effects and spectral contributions we generate synthetic time series data and apply the developed formulas to them. 

Section \ref{sec: asymptotic efficiency} is motivated by a recent result of \citep{Guo2020EfficientLS}, where for causally sufficient linear models the authors identify among all consistent and differentiable causal effect estimators the one with the lowest asymptotic variance. We generalise this result to estimators for frequency domain causal quantities in time series models. 

In Section \ref{sec: application}, we use our characterisations of the asymptotic distribution to re-examine if and how anomalies in solar activity affect anomalies in the Northern Atlantic Oscillation (NAO), a question that has been intensely studied and debated in the climate sciences \citep{gray2010solarinfluenceclimate, gray2013responsesolarcycle, gray2016solarcycleNAO, Characterizationofthe11YearSolarSignalUsingaMultipleRegressionAnalysisoftheERA40Dataset}. In our analysis, we focus on frequencies that represent oscillations that last about 10 years, as the variability in solar activity is concentrated on this time scale. Our Wald type test suggests with 95 $\%$ confidence that variability in solar activity contributes to variations of the NAO on this time scale. 

\section{Preliminaries}\label{sec: preliminaries}
\subsection{Graphical modelling} Suppose $V$ is a finite set that indexes the processes to be modeled. A \emph{time series graph} for the processes $V$ is a \emph{directed acyclic graph (DAG)} $\G = (V \times \Z, \D)$, so that a vertex $v(t)$  on $\G$ represents the state of the process indexed by $v$ at time $t$. The set of edges $\D$ must be such that for any link $(u(t-k), v(t))\in \D$ it holds that $k \geq 0$ and also that $\{(u(s-k), v(s)): s \in \Z) \} \subset \D$, which in the following we will denote by $u \to_k v $. The \emph{order} of the time series graph $\G$ is the maximum time lag with one process is causing another, that is $p \coloneqq \max_{u,v\in V}\{k : u \to_k v\}$. We denote by $\Lag^\G_v$ the set of all $(u,k)$ such that $u \to_k v$. The \emph{contemporaneous DAG} $\G_0 = (V, \D_0)$ contains a direct link from $u$ to $v$ if and only if there is a contemporaneous effect $u\to_0 v$ on $\G$. Consequently, the set of contemporaneous parents of $v\in V$, written $\pa_0(v)$, consists of all such processes $u$ that satisfy $u \to_0 v$.
The \emph{process graph} of $\G$ is a finite graph $G=(V,D)$ over the set of process indices $V$. This graph contains a directed link from $u$ to $v$, denoted $u \to v$, if and only if $u \neq v$ and $u\to_k v$ for some $k \geq 0$. The parent set of a process $v \in V$ is $\pa(v) \coloneqq \{u \in V : u \to v \}$. In this work we will always assume that the structure of the process graph $G$ and of the contemporaneous graph $\G_0$, together with an upper bound $q \geq p$ for the order of $\G$ is known. However, the details of the time series graph in form of the sets $\Lag_v^\G$ that go beyond the structure of the contemporaneous graph may not be known. The information on $\G$ given $G$ and $\G_0$ and an upper bound $q \geq p$ can be encoded by selecting for every process $v\in V$ a finite set of time lagged relations $\Lag_v$ such that 
\begin{align}\label{condition: time lags}
	 \Lag_v^\G \subset  \Lag_v &\subset \pa_0(v)  \times \{0\} \cup (\{v\} \cup \pa(v)) \times [1,q].
\end{align} 
Throughout this work we use the symbol $\Lag $ to refer to a collection of time lagged relations $\{\Lag_v\}_{v \in V}$ such that each $\Lag_v$ satisfies (\ref{condition: time lags}). Accordingly, $\Lag^\G$ denotes the collection $\{\Lag^\G_v\}_{v \in V}$. 
We illustrate the relation between a time series graph and its process  graph with an example in Figure \ref{fig: preliminaries tsg}. 
\begin{figure}
    \centering
    \includegraphics[width=0.4\textwidth]{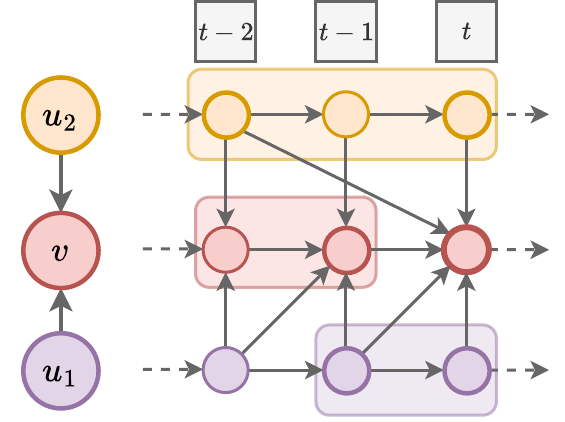}
    \caption{A process graph (left) together with an excerpt of its underlying time series graph (right).}
    \label{fig: preliminaries tsg}
\end{figure}
To avoid overly complicated technicalities, we will often need to make the following assumption in this work. 
\begin{assumption}[AL2]\label{assumption: AL2}
    A set of time lagged relations $\Lag_v$ is said to satisfy the AL2 assumption if for every $u \in V$ if the set $\{j \mid \in (u,j) \in \Lag_v \}$ is either empty or contains at least two elements. 
\end{assumption}
\subsection{Structural vector autoregressive process processes}
A \emph{structural vector autoregressive (SVAR) process} for the time series graph $\G$ of order $p$ is a multivariate discrete time stochastic process $\X = (\X_v(t))_{v \in V, t \in \Z}$. The process $\X$ is specified by a multivariate Gaussian white noise process $\eta = (\eta_v(t))_{v\in V, t \in \Z}$, i.e. the variables $\eta_v(t)$ are mutually independent Gaussians for all $v \in V$ and $t\in \Z$, 
and a finite coefficient vector $\Phi = (\phi_{u,v}(k))_{u\to_k v} \in \R^{\Lag^\G}$. The dimensions of $\Phi$ are indexed by the direct lagged and contemporaneous causal relations. The vector $\Phi$ and the noise process $\eta$ must be such that 
\begin{equation}
    \X_v(t) = \eta_{v}(t) + \sum_{(u,k) \in \Lag_v^\G} \phi_{u,v}(k) \X_u(t-k)
\end{equation}
for every $v\in V$. Throughout this work we require every SVAR coefficient vector to be such that the process $\X$ is \emph{stable} \citep{lutkepohl2005new, brockwell2009time} and condition (\ref{assimption: stability}) in the supplementary material. With the finite set of time lagged relations $\Lag_v$ we associate a parameter vector $\Phi^{\Lag_v} =(\phi_{u,v}^{\Lag_v}(k))_{(u,k) \in \Lag_v}$ and a stochastic process $\X^{\Lag_v}$, i.e.
\begin{align*}
    \phi^{\Lag_v}_{u,v}(k) &\coloneqq \begin{cases}
        \phi_{u,v}(k) & \text{if } (u,k) \in \Lag_v^\G \\
        0 & \text{if } (u,k) \in \Lag_v \setminus \Lag_v^\G
    \end{cases} &
    \X^{\Lag_v}(t) &\coloneqq \begin{bmatrix}
        \X_{u}(t-k)
    \end{bmatrix}_{(u, k) \in \Lag_v}.
\end{align*}
\begin{example}
	Let us illustrate the notation with the example depicted in Figure \ref{fig: preliminaries tsg}. For process $v$ the true set of time lagged causal effects is 
     \begin{align*}
         \Lag_v^\G = \{(u_1, 0), (u_1, 2), (v, 1), (u_2, 0), (u_2, 1)\}.
     \end{align*} 
	In Figure \ref{fig: preliminaries tsg} we represent the set of considered time lags $\Lag_v$ by the nodes inside the boxes in the time series graph. 
	The choice of the orange box in the row corresponding to process $u_2$ reflects the information that the time lag with which $u_1$ is driving $v$ is not larger than two. Similarly, the box in the row of $u_2$ encodes the assumption that the time lag with which $u_2$ is causing $v$ is at most one, and the time lag with which $v$ is driving itself is at most two (red box in the second row). So the coefficient vector and stochastic process associated with the lags $\Lag_v$ are as follows
	\begin{align*}
		\Phi^{\Lag_v} &= \begin{bmatrix}  \phi_{u_1, v}(0) & \phi_{u_1, v}(1)   \phi_{u_2,v}(0) & 0 & \phi_{u_2,v}(2)  & \phi_{v, v}(1) & 0 \end{bmatrix} \\
		\X^{\Lag_v}(t) &= \begin{bmatrix} \X_{u_1}(t) & \X_{u_1}(t-1) & \X_{u_2}(t) & \X_{u_2}(t-1) & \X_{u_2}(t-2) &  \X_{v}(t-1) & \X_{v}(t-2) \end{bmatrix}.
	\end{align*}
\end{example}

Since $\Lag_v$ satisfies (\ref{condition: time lags}) it holds that $\X_v(t) = (\Phi^{\Lag_v})^\top \X^{\Lag_v}(t) + \eta_v(t)$, and the parameter $\Phi^{\Lag_v}$ satisfies the regression relation
\begin{equation}\label{eq: regression relation}
    \Phi^{\Lag_v} = \E[(\X^{\Lag_v}(t)) (\X^{\Lag_v}(t))^\top]^{-1} \E[(\X^{\Lag_v}(t) )(\X_v(t))].
\end{equation}
Suppose we have a time series of length $T+q$, where $q\geq p$ is the maximum time lag occurring in $\Lag$, then one can estimate the above covariance information and use relation (\ref{eq: regression relation}) to obtain the ordinary least squares (OLS) estimator $\hat{\Phi}^{\Lag_v}$ for the estimand $\Phi^{\Lag_v}$. For large $T$ the estimator $\hat{\Phi}^{\Lag_v}$ is approximately normally distributed, i.e.  
\begin{align}\label{eq: asymptotic normality phi}
    \sqrt{T}(\hat{\Phi}^{\Lag_v}- \Phi^{\Lag_v}) &\to_d \mathcal{N}(\mathbf{0}, P^{\Lag_v}) & P^{\Lag_v} &= \omega_v \E[\X^{\Lag_v}(t) (\X^{\Lag_v}(t))^\top]^{-1}.
\end{align}

\subsection{Structural causal models of stochastic processes}
A SVAR process admits an equivalent formulation as a \emph{structural causal model of stochastic processes} at the level of its process graph $G$ \citep{reiter2023formalising}. This model consists of a set of stochastic processes $\X^\internal = (\X^\internal_v)_{v \in V}$ and a \emph{link filter} $\lambda_{u,v}= (\lambda_{u,v}(s))_{s\in \Z}$ for every link $u\to v$ on the process graph $G$ such that 
\begin{equation}\label{eq: direct link filters}
    \X_v = \X^\internal_v + \sum_{u \in \pa(v)} \lambda_{u,v} \ast \X_u,
\end{equation}
where $\ast$ denotes the convolution of two sequences. Infinite sequences and convolution operation on them do not seem practical for computations. However, the computations become more tractable in the frequency domain after application of the Fourier transform $\F$. In order to compactly express the Fourier transformed link filters in terms of the underlying SVAR parameter $\Phi$ let us introduce for any two $x,y \in V$ their associated \emph{lag polynomial}
\begin{equation}\label{eq: lag polynomial}
    \varphi_{x,y}(z) \coloneqq \sum_{(x,k) \in \Lag_y}\phi_{x,y}^{\Lag_y}(k)z^k = \sum_{(x,k)\in \Lag^\G_y}\phi_{x,y}(k)z^k,
\end{equation}
evaluated at frequency $z \in S^1 = \{z \in \C: |z| =1 \}$. 
The \emph{link function} \cite{reiter2023formalising} $\fsep_{u,v}$ of $u \to v$ is the Fourier transform of the link filter $\lambda_{u,v}$. This function is defined on the complex unit circle $S^1$ and parameterised as a fraction of lag polynomials, that is 
\begin{equation}\label{eq: link function}
    \fsep_{u,v}(z) \coloneqq (\F(\lambda_{u,v}))(z) = \frac{\varphi_{u,v}(z)}{1 - \varphi_{v,v}(z)}.
\end{equation}

\section{Asymptotic distribution of causal effect estimators in the frequency domain}\label{sec: asymptotic distribution}

\subsection{Asymptotic uncertainty in the estimation of link functions}
For the rest of this section, we fix a time series graph DAG $\G$ with order $p \geq 0$ and process graph $G=(V,D)$. Our information on $\G$ relative to the process graph $G$ and the contemporaneous graph $\G_0$ is encoded by a collection of time lagged relations $\Lag = \{ \Lag_v \}_{v \in V}$. 
In order to quantify the uncertainty in the estimation of link functions we would like to consider them as real valued functions in the parameters of the SVAR model. So from now on we read a given link function as a two dimensional real valued function of which the first coordinate is the real part and the second is the imaginary part. To express the function for the real resp. imaginary part as a rational function in the SVAR coefficients we make use of the fact that the non-zero complex numbers are isomorphic to a particular subgroup of the orthogonal $2\times 2$ dimensional matrices. Specifically, if $0 \neq z \in \C$ is a complex number, then we define its corresponding operator as 
\begin{align}\label{eq: complex operator}
	M(z) &= \begin{bmatrix}\re(z) & -\imag(z) \\ \imag(z) & \re(z) \end{bmatrix} & M(z^{-1}) = M(z)^{-1} &= \frac{1}{|z|^2} \begin{bmatrix}\re(z) & \imag(z) \\ -\imag(z) & \re(z) \end{bmatrix} 
{}\end{align}  
Note that $M(z)e_1 = z$ and $M(zz') = M(z)M(z')$ and $M(z^\ast) = M(z)^\top$. Furthermore, for $x, y \in V$ the evaluation of the real and imaginary part of the lag polynomial $\varphi_{x,y}$ are polynomials in the coefficients $\Phi^{\Lag_y}$, i.e. 
\begin{align}\label{eq: real imag lag polynomial}
	\re(\varphi_{x,y}(z)) &= \sum_{(x,k) \in \Lag_y} \phi_{x,y}^{\Lag_y}(k)\re(z^k) & \imag(\varphi_{x,y}(z)) &= \sum_{(x,k) \in \Lag_y} \phi_{x,y}^{\Lag_y}(k)\imag(z^k).
\end{align} 
By combining (\ref{eq: complex operator}) and (\ref{eq: real imag lag polynomial}) we conclude that the real and imaginary part of the link function are rational functions in the parameters $\Phi^{\Lag_v}$, i.e., 
\begin{align*}
	\re(\fsep_{u,v}(z)) &= \frac{\re(\varphi_{u,v}(z))\re(1- \varphi_{u,v}(z)) - \imag(\varphi_{u,v}(z))\imag(\varphi_{v,v}(z))}{\re(1- \varphi_{v,v}(z))^2 + \imag(\varphi_{w,w}(z))^2} \\
	\imag(\fsep_{u,v}(z)) &= \frac{\re(\varphi_{u,v}(z)) \imag(\varphi_{v,v}(z)) + \imag(\varphi_{u,v}(z))\re(1-\varphi_{v,v}(z))}{\re(1- \varphi_{v,v}(z))^2 + \imag(\varphi_{v,v}(z))^2}
\end{align*}
\begin{example}
 	We use the example depicted in Figure \ref{fig: preliminaries tsg} to demonstrate that the real and imaginary part of the link function $\fsep_{u_1, v}(z)$  are rational functions in the SVAR coefficients. 
	For the sake of easier readability we abbreviate $x_k = \phi^{\Lag_v}_{u_1, v}(k)$ and $y_j = \phi^{\Lag_v}_{v,v}(j)$ for $k=0,1$ and $j=1,2$. We compute the rational functions for the real and the imaginary part
	\begin{align*}\re(\fsep_{u_1, v}(z)) &= \frac{p(x_0, x_1, y_1, y_2)}{q(y_1, y_2)} & 
	\imag(\fsep_{u_1, v}(z))& = \frac{p'(x_0, x_1, y_1, y_2)}{q(y_1, y_2)}.\end{align*} The denominator $q$ is a polynomial in the two variables $y_1, y_2$, i.e. 
	\begin{align*}
		q(y_1, y_2) &= 1- \mu_1 y_1 - \mu_2 y_2  + \mu_{1,1}y_1^2 + \mu_{2,2}y_2^2 + \mu_{1,2}y_1y_2
	\end{align*}
	with coefficients $\mu_i = 2 \re(z^i)$ and $\mu_{i,i} = \re(z^i)^2 + \imag(z^i)^2$ for $i=1,2$, and $\mu_{1,2} =2 (\re(z)\imag(z^2) + \re(z^2)\imag(z))$. Similarly, the numerator of the real resp. imaginary part is a polynomial of degree two in four variables, i.e. 
	\begin{align*}
		p(x_0, x_1, y_1, y_2) &= \alpha_0 x_0 + \alpha_1 x_1 + \alpha_{0,1} x_0y_1 + \alpha_{0,2 }x_0y_2 + \alpha_{1,1}x_1y_1 + \alpha_{1,2}x_1y_2, \\
		p'(x_0, x_1, y_1, y_2) &= \beta_0 x_0 + \beta_1 x_1 + \beta_{0,1} x_0y_1 + \beta_{0,2 }x_0y_2 + \beta_{1,1}x_1y_1 + \beta_{1,2}x_1y_2
	\end{align*}
	with coefficients $\alpha_i = \re(z^i)$ and $\beta_i = \imag(z^i)$, and $\alpha_{i,j} = \imag(z^i)\imag(z^j)-\re(z^i) \re(z^j)$ and $\beta_{i,j} = \re(z^j)\imag(z^i) - \re(z^i)\imag(z^j)$ for $i=0,1$ and $j=1,2$. 
\end{example} 
For a link $u\to v$ on the process graph, we write $\hat{\fsep}_{u,v}(z)$ to denote the estimator for $\fsep_{u,v}(z)$ obtained by plugging the OLS estimator $\hat{\Phi}^{\Lag_v}$ into the equations (\ref{eq: lag polynomial}) and (\ref{eq: link function}). Thus, the uncertainty in $\hat{\fsep}_{u,v}(z)$ is determined by the uncertainty in the asymptotically normal OLS estimator $\hat{\Phi}^{\Lag_v}$. From the delta method \citep{Vaart_1998} it follows that the estimator $\hat{\fsep}_{u,v}(z)$ follows a two-variate normal distribution if the gradient of the link function as a function of $\Phi^{\Lag_v}$ has full rank. If $\Lag_v$ satisfies the assumption \ref{assumption: AL2}, then the gradient has full rank for generic choices of $\Phi$, which means that the set of parameters $\Phi$ for which the rank is not full has Lebesgue measure zero. Furthermore, all the link function estimators for the links pointing to the process vertex $v$ and the estimator for the \emph{internal function} of $v$, which we define as $\intsep_v \coloneqq (1- \varphi_{v,v})^{-1}$, asymptotically follow a joint normal distribution. Let us write $\Fsep_{\bullet, v}(z)$ to denote the vector composed of the internal function at $v$ and the link functions associated with the parents of $v$, each evaluated at $z \in S^1$. Even though we do not need to use the internal functions yet, we will introduce them at this point as it will be useful later on.
\begin{proposition}\label{prop: asymptotic normality link functions} 
Suppose $\Lag_v$ satisfies the AL2 assumption \ref{assumption: AL2}. Then for generic choices of SVAR coefficients $\Phi$ and for all but finitely many $z\in S^1$, the asymptotically normal OLS estimator $\hat{\Fsep}_{\bullet, w}(z)$, i.e.
\begin{align}\label{eq: asymptotic cov link functions} 
        \sqrt{T}(\hat{\Fsep}_{\bullet,v}(z)- \Fsep_{\bullet,v}(z)) &\to_d \mathcal{N}(\mathbf{0}; \begin{bmatrix}
        \acov^{\Lag_v}(u_1,u_2;z)
    \end{bmatrix}_{u_1, u_2 \in \pa(v)\cup \{v\}})
\end{align}
has positive definite asymptotic covariance matrix. 
The asymptotic covariance matrix is block structured into $2\times 2$-matrices indexed by tuples of $\pa(v) \cup \{v\}$, and the $2\times 2$ block matrix $\acov^{\Lag_v}(u_1,u_2;z)$, i.e. the asymptotic covariance of $\hat{\fsep}_{u_1,v}(z)$ and $\hat{\fsep}_{u_2,v}(z)$, is a rational function in the SVAR parameters $\phi_{u_1,v}(k), \phi_{u_2,v}(j), \phi_{v,v}(l)$. 
\end{proposition}
In the appendix we provide the explicit rational functions for the block entries in the asymptotic covariance $\acov^{\Lag_v}(z)$. 
\subsection{The asymptotic covariance of the causal effect estimator}
Viewing an SVAR process as a structural causal model of stochastic processes enables one to quantify causal effects between processes at the level of the process graph. These process level effects can be computed by means of a generalized path rule.
A \emph{directed path} between processes is a sequence of consecutive edges on the process graph. In particular, we allow paths to visit a vertex more than once. In this work, we treat the empty path at any process node $v$ as a directed path and denote it as $\epsilon_v$. Suppose $v,w \in V$ are process indices, then we write $\paths(v,w)$ for the set of all directed paths starting at $v$ and ending at $w$. With the direct link filters (\ref{eq: direct link filters}) it is possible to systematically quantify the influence the processes have on another. Specifically, if $\Pi \subset \paths(v,w)$, then the part of $\X_w$ that is determined by the process $\X_v$ along all the paths $\Pi$ is identified by the formula
\begin{align}\label{eq: path wise contribution}
    \X^\Pi &= (\sum_{\pi \in \Pi} \lambda^{(\pi)}) \ast \X_v, &\lambda^{(\pi)} &= \prod_{(x,y) \in \pi} \lambda_{x,y},
\end{align}
where $\lambda^{(\pi)}$ refers to the path filter of $\pi$, which is the convolution of all the link filters on $\pi$, see \cite{reiter2023formalising}. This description quantifies the causal structure among the processes of $\X$ in terms of the structure of the process graph $G$, and generalises the classical path rule \citep{wright1934method}. Once more, the expressions are easier to compute in the frequency domain, i.e. 
\begin{align}\label{eq: path function}
     \fsep^\Pi &\coloneqq \F(\sum_{\pi \in \Pi} \lambda^{(\pi)}) = \sum_{\pi \in \Pi} \fsep^{(\pi)} & \fsep^{(\pi)} &\coloneqq \F(\lambda^{(\pi)}) = \prod_{(v,w)\in \pi} \fsep_{v,w},
\end{align} 
where $\fsep^{(\pi)}$ is the \emph{path function} \cite{reiter2023formalising} associated with the path $\pi$, which is the point-wise multiplication of the link functions associated with the links on $\pi$, and $\fsep^\Pi$ is the \emph{total effect function} for the set of paths $\Pi$. In particular, the path function of the empty path is the constant function. The evaluation of a path or total effect function at some frequency $z$ is a rational function in the SVAR parameter $\Phi^\Lag = (\Phi^{\Lag_v})_{v \in V}$. This follows because every link function is a rational function in $\Phi^\Lag$ and by using the identities below (\ref{eq: complex operator}). In the following, we write $\hat{\fsep}^{(\pi)}$ to refer to the OLS-based estimator for $\fsep^{(\pi)}$ that we obtain by inserting the OLS result $\hat{\Phi}^\Lag$ into (\ref{eq: path function}).

In the following proposition we give an expression for the asymptotic covariance of OLS-based path function estimators. To keep the notation simple, we consider only directed paths $\pi$ that visit each process vertex of $G$ at most once, i.e. paths that do not contain cycles.
\begin{proposition} \label{prop: asymptotic covariance path function}
    Suppose $\pi$ and  $\rho$ are two directed paths without repeating vertices on the process graph $G$ with starting points $v$ resp. $v'$. Then the asymptotic covariance between the estimators $\hat{\fsep}^{(\pi)}$ and $\hat{\fsep}^{(\rho)}$ is 
    \begin{align*}
        \acov^\mathbf{L}(\pi, \rho) &= \sum_{v \in V(\pi \cap \rho)\setminus\{v,v'\}} \fsep^{(\pi\setminus v)} [\acov^{\Lag_v}(\pi(v), \rho(v))] (\fsep^{(\rho\setminus v)})^\ast, & \fsep^{(\pi\setminus v)} &\coloneqq \prod_{x \to y \in \pi : y \neq v} \fsep_{x,y}
    \end{align*}
    where $V(\pi \cap \rho)$ is the set of all vertices that are visited by both $\pi$ and $\rho$, and if $v$ is a vertex visited by $\pi$, then $\pi(v)\in V$ is the vertex that precedes $v$ on the path $\pi$. 
\end{proposition}
In particular, the asymptotic covariance of two path functions is zero if the paths do not intersect. 
\begin{figure}
    \centering
    \includegraphics[width=0.5\textwidth]{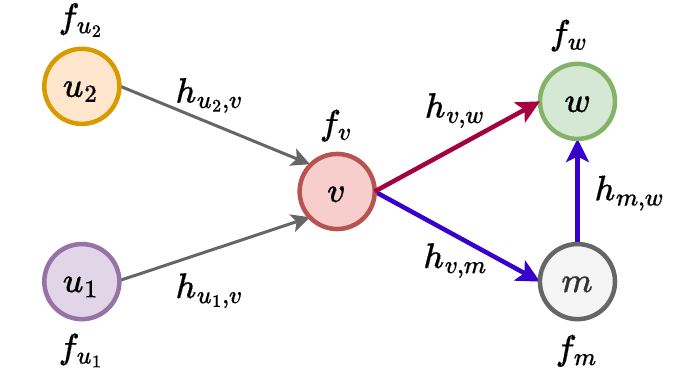}
    \caption{A process graph representing causal interactions among the processes $V= \{u_1, u_2, v, m, w\}$. Each vertex is annotated by the function generating the internal dynamics, and each edge is annotated by the associated link function.}
    \label{fig:process graph}
\end{figure}
Before moving to the next subsection, we introduce two more notations. Suppose $\Pi_1$ and $\Pi_2$ are sets of directed paths, then we define the asymptotic block covariance matrix for the path function estimators as 
\begin{align*}
    \acov^\Lag(\Pi_1, \Pi_2) \coloneqq \begin{bmatrix}
        \acov^\Lag(\pi_1, \pi_2)
    \end{bmatrix}_{\pi_1 \in \Pi_1, \pi_2 \in \Pi_2}.
\end{align*}
The asymptotic covariance of the total effect estimators $\hat{\fsep}^{\Pi_1}$ and $\hat{\fsep}^{\Pi_2}$ is  
\begin{align}
    \acov^\Lag(\hat{\fsep}^{\Pi_1}, \hat{\fsep}^{\Pi_2})\coloneqq \sum_{\pi_1 \in \Pi_1} \sum_{\pi_2 \in \Pi_2} \acov^\Lag(\pi_1, \pi_2).
\end{align}
\begin{example}
	For the process graph depicted in Figure \ref{fig:process graph} we illustrate the formula with which to compute the asymptotic covariance of the estimators for the causal effect of the process $v$ on the process $w$. There are two directed paths along which the process $v$ could possibly impact the process $w$. So the asymptotic covariance $\acov^{\Lag}(\Pi, \Pi)$ is a $2\times 2$ block matrix and each pair of paths from $v$ to $w$ indexes a block entry. Using Proposition \ref{prop: asymptotic covariance path function} we obtain the following expressions for the block entries 
\begin{equation} \label{eq: example acov path functions}
    \begin{split}
    \acov^\Lag(\pi_{v,w}, \pi_{v,w}) &= \acov^\Lag_w(v,v) \\
    \acov^\Lag(\pi_{v,w}, \pi) &=\acov^\Lag_w(v,m)(\fsep_{v,m})^\ast \\
    \acov^\Lag(\pi, \pi) &= \fsep_{v,m}\acov^\Lag_w(m,m)(\fsep_{v,m})^\ast + \fsep_{m,w}\acov^\Lag_m(v,v)(\fsep_{m,w})^\ast
    \end{split}
\end{equation}
\end{example}
\subsection{The asymptotic covariance of the spectral contribution estimator}
Path and total effect functions provide a frequency domain description of how one process $w \in V$ responds to variations in another process $v \in V$ through a set of directed paths $\Pi \subset \paths(v,w)$. To get a better understanding of what these functions quantify let us consider the \emph{spectral density} of the SVAR process $\X$ and its internal dynamics $\X^\internal$, which we denote by $\sd$ and $\sd^\internal$ respectively. The spectral density (of the internal dynamics) is the Fourier transform of the \emph{autocovariance (ACS)} $\Sigma = (\Sigma(k))_{k \in \Z}$ and $\Sigma^\internal$ respectively, which are defined as follows
\begin{align*}
    \Sigma(k) &\coloneqq \E[\X(t) \X(t-k)^\top] & \Sigma^\internal(k) &\coloneqq \E[\X^\internal(t) \X^\internal(t-k)].
\end{align*}
The spectral density is a statistical measure of the dynamics of the SVAR process $\X$ derived from the \emph{spectral representation} \cite{brockwell2009time} of $\X$. The spectral representation of $\X$ is a superposition of oscillations, i.e. infinitely recurring temporal patterns, and each oscillation is characterised by its frequency, which encodes the time it takes the oscillation to complete a cycle. The complex-valued amplitude of each oscillation is a random variable, and the amplitudes of two different oscillations are independent of each other. The spectral density $\sd_{w}(z)$ measures the variance of the amplitude associated with the oscillation at frequency $z$. The spectral density thus provides information on how the variation of the component process $\X_w$ is distributed across time scales. The fraction of $\sd_w(z)$ determined by the process $\X_v$ along the set of paths $\Pi$ is quantified by the path or total effect functions. Specifically, this fraction is the spectral density of the process $\X^\Pi$ as defined in (\ref{eq: path wise contribution}). We call this fraction the \emph{spectral contribution} of $v$ to $w$ along $\Pi$ and denote it by $\sd^\Pi$.
The spectral density $\sd^\Pi$ admits a rational expression in terms of the SVAR parameters $\Phi^\Lag$. To see this, we observe that the definition of $\X^\Pi$ implies that $\sd^\Pi = |\fsep^\Pi|^2\sd_v$. Since $|\fsep^\Pi|$ is a rational function in $\Phi^\Lag$, it remains to find a rational parameterisation of $\sd_v$. To do this, we recall that the set of ancestors of $v$, written $\anc(v)\subset V$, consists of those vertices $u \in V$ that satisfy $P(u,v)\neq \emptyset$, and due to our convention regarding empty paths, we consider each vertex $v$ as an ancestor of itself, i.e. $v\in \anc(v)$. Then, using the frequency domain trek rule \citep{reiter2023formalising}, we get that
\begin{align*}
	\sd_v &= \sum_{u \in \anc(v)} |\fsep^{\paths(u,v)}|^2 \sd_u^\internal & \sd^\internal_u(z) &= \omega_u|\intsep_{u}(z)|^2  \\
	&= \sum_{u \in \anc(v)} |\fsep^{\paths(u,v)} \fsep_{u,u}|^2 & &= \omega_u |1- \varphi_{u,u}(z)|^{-2}
\end{align*}
This shows that $\sd^\Pi$ is a rational function of the SVAR parameter. 
If $\pi\in \paths(u,v)$ is a path, then we refer to the function $\csep^{(\pi)} \coloneqq \fsep^{(\pi)}\intsep_{v}$ as the \emph{weighted path function} of $\pi$, and for $\Pi \subset \paths(v,w)$ we denote by $\csep^{\Pi}$ the sum of all weighted path functions of the paths in $\Pi$. In particular, the weighted path function of the empty path $\epsilon_v$ at $v$ is the internal function of $v$, i.e. $\csep^{(\epsilon_v)} = \intsep_v$. We conclude that $\sd^\Pi$, i.e. the spectral density of $\X^\Pi$, may be written in terms of weighted path functions
\begin{align}\label{eq: spectral contribution}
    \sd^{\Pi} \coloneqq |\fsep^\Pi|^2 \sd_v 
    = \sum_{u \in \anc(v)} \omega_u |\csep^{\Pi_u}|^2 ,
\end{align}
where $\Pi_u = \paths(u,v) + \Pi= \{ \rho_u + \pi : \rho_u \in \paths(u,v), \pi \in \Pi \}$ is the set of all possible concatenations. The function $|\csep^{\Pi_u}|^2$ quantifies the \emph{spectral contribution} of $u$ via $v$ on $w$ along $\Pi$, up to the amplitude $\omega_u$. We now wish to determine the asymptotic uncertainty of the estimator $\hat{\sd}^\Pi$, which is $\hat{\Phi}^\Lag$ plugged into (\ref{eq: spectral contribution}). Since we expressed this contribution in terms of weighted path functions, we begin with an expression for the asymptotic covariance between any two weighted path function estimators. 
\begin{proposition} \label{prop: asymptotic covariance weighted path function}
    Suppose $\pi$ and $\rho$ are two directed paths without cycles on the process graph $G$, where $\pi$ is going from $v$ to $w$ and $\rho$ from $v'$ to $w'$. The asymptotic covariance between the estimator $\hat{\csep}^{(\pi)}$ and $\hat{\csep}^{(\rho)}$ is 
    \begin{align*}
        \acov^\Lag(\pi^\internal, \rho^\internal) &= \begin{cases}
        \intsep_{v}[\acov^\Lag(\pi, \rho)](\intsep_{v'})^\ast & \text{, if } v \neq v' \\
        \intsep_{v}[\acov^\Lag(\pi, \rho)](\intsep_{v})^\ast + \fsep^{(\pi)} [\acov^{\Lag_v}(v,v) ]\left(\fsep^{(\rho)}\right)^\ast & \text{, if } v = v'
        \end{cases}
    \end{align*}
\end{proposition}
To quantify the asymptotic uncertainty of the spectral contribution estimator, we need to determine for any two $u,u'\in \anc(v)$ the asymptotic covariance $\acov^\Lag(\hat{\csep}^{\Pi_u}, \hat{\csep}^{\Pi_{u'}})$, which we express in the following. 
\begin{proposition}\label{prop: asymptotic covariance spectral contribution}
    Let $v\in V$ be such that for every $u \in \anc(v)$ the set of paths $\Pi_u$ is finite. Then for $u,u'\in \anc(v)$ it holds that   
    \begin{align*}
        \acov^\Lag(\hat{\csep}^{\Pi_u}, \hat{\csep}^{\Pi_{u'}}) &= \fsep^\Pi [\acov^\Lag(\hat{\csep}^{\paths(u,v)}, \hat{\csep}^{\paths(u',v)})] (\fsep^\Pi)^\ast +\csep^{\paths(u,v)} [\acov(\hat{\fsep}^\Pi)] (\csep^{\paths(u',v)})^\ast
    \end{align*}
\end{proposition}
We give the proof in the supplementary material. 
In the next subsection, we will use Proposition \ref{prop: asymptotic covariance spectral contribution} to derive an asymptotic confidence region for the estimator $\hat{\sd}^{\Pi}$.
\begin{example}
	For the example process graph shown in Figure \ref{fig:process graph} we use Proposition \ref{prop: asymptotic covariance weighted path function} and \ref{prop: asymptotic covariance spectral contribution} to spell out the asymptotic covariance for the estimation of the spectral contribution of $v$ to  the process $w$, which we identified as $|(\fsep_{v,w} + \fsep^{(\pi)})|^2 \sd_v$. First, we need to derive the expressions for the asymptotic covariance of the spectral density of the process $v$, which is computed as follows 
\begin{align*}
    \sd_v 
    &= |\csep^{(\pi_{u_1, v})}|^2 + |\csep^{(\pi_{u_2, v})}|^2 + |\csep^{(\epsilon_{v})}|^2 .
\end{align*}
According to Proposition \ref{prop: asymptotic covariance spectral contribution} we need to compute the asymptotic block covariance of the multivariate estimator of weighted path functions $(\hat{\csep}^{(\pi_{u_1, v})}, \hat{\csep}^{(\pi_{u_2, v})}, \hat{\csep}^{(\epsilon_{v})})$. Using the notation from the previous section and Proposition \ref{prop: asymptotic covariance weighted path function} we get that the block entries on the diagonal are as follows 
\begin{equation} \label{eq: example diagonal acov weigthed path functions}
    \begin{split}
    \acov^\Lag(\pi_{u_1, v}^\internal, \pi_{u_1,v}^\internal) &= \intsep_{u_1}\acov^{\Lag_v}(u_1, u_1)\intsep_{u_1}^\ast \\
    \acov^\Lag(\pi_{u_2, v}^\internal, \pi_{u_2,v}^\internal) &= \intsep_{u_2}\acov^{\Lag_v}(u_2, u_2)\intsep_{u_2}^\ast \\
    \acov^\Lag(\epsilon_{v}^\internal, \epsilon_{v}^\internal) &= \acov^{\Lag_v}(v, v)
    \end{split}
\end{equation}
The expressions for the off diagonal entries are computed analogously (see the supplement \ref{sec: supplement numerical example}). The estimator for the spectral contribution of $u_1$ along $\Pi$ to $w$ is $\hat{\csep}^{\Pi_{u_1}} = \hat{\csep}^{(\pi_{u_1, v})} \hat{\fsep}^{\Pi}$. Proposition \ref{prop: asymptotic covariance spectral contribution} identifies its asymptotic covariance as 
\begin{align*}
    \acov^{\Lag}(\hat{\csep}^{\Pi_{u_1}}, \hat{\csep}^{\Pi_{u_1}}) &= (\fsep^{\Pi}\intsep_{u_1})\acov^{\Lag_v}(u_1, u_1)(\fsep^{\Pi}\intsep_{u_1})^\ast + (\fsep_{u_1,v} \intsep_{u_1}) \acov^\Lag(\hat{\fsep^{\Pi}}, \hat{\fsep^{\Pi}}) (\fsep_{u_1,v} \intsep_{u_1})^\ast
\end{align*}
The remaining blocks in the asymptotic covariance of the estimator for the spectral contribution are shown in the supplement \ref{sec: supplement numerical example}. 
\end{example}

\subsection{Asymptotic significance tests and confidence regions}\label{subsec: asymptotic uncertainty and significance}
A process graph $G$ expresses a collection of hypotheses about the process $\X$. Namely, each path $\pi \in \paths(v,w)$ on the process graph implicitly hypothesises that the behaviour of process $v$ influences the behaviour of process $w$. Since these hypotheses can be quantified using (weighted) path functions, they can be tested on observational data. In this subsection, we construct asymptotic hypothesis tests to assess whether observational time series data of length $T+q$ are consistent with the hypothesis that process $v$ causally drives $w$ along a set of paths $\Pi$ at a given significance level $\alpha \in (0,1)$. We saw earlier that the estimators for (weighted) path functions asymptotically follow a normal distribution, and the expressions for their asymptotic covariance matrices admit explicit rational expressions in terms of $\Phi^\Lag$. 

Hypotheses about asymptotically normal estimators can be tested using what is known as the Wald test \cite{wald1943statistics} or $\chi_2$-test \cite{Vaart_1998}. Specifically, suppose we are interested in a vector valued quantity $\theta \in \R^{m}$ for which we have an asymptotically normal estimator $\hat{\theta}_n$, i.e. $\sqrt{n}(\hat{\theta}_n - \theta) \to_d \mathcal{N}(\mathbf{0}; \Sigma)$, where $n$ is the number of samples from which we construct the estimator, and $\Sigma$ the asymptotic covariance matrix. We now want to decide whether the true but unknown value $\theta$ is different from the zero vector at some significance level $\alpha$. To do this, we assume the opposite and take this as the null hypothesis, i.e. $H_0: \theta = \mathbf{0} \iff \lVert \theta \rVert = 0$. Consequently, the alternative hypothesis becomes $H_1: \theta \neq \mathbf{0} \iff \lVert \theta \rVert > 0$. Under the null hypothesis $H_0$ it holds that 
\begin{align}\label{eq: general wald statistic}
	W_n \coloneqq n \hat{\theta}_n^\top  \hat{\Sigma}^{-1} \hat{\theta}_n \to_d \chi^2(m).
\end{align}
This means that under $H_0$ and for large sample size $n$, the Wald statistic $W_n$ approximately follows the distribution of a $\chi^2(m)$ distributed random variable. The Wald test rejects $H_0$ if $\Prb(\chi^2(m) \geq W_n)) < 1 - \alpha$, which is equivalent to $W_n \geq x_{\alpha; m}$, where $x_{\alpha; m}$ is the $\alpha$-quantile of the $\chi^2(m)$ distribution. One can use convergence (\ref{eq: general wald statistic}) to derive an approximate confidence region for the estimator $\hat{\theta}_n$ and significance level $\alpha$. Specifically, for the elliptical set 
\begin{align}\label{eq: asymptotic confidence region}
    C_{\alpha} = \left\{ \mathbf{x}\in \R^m \mid n(\mathbf{x}-\theta)^\top\Sigma^{-1}(\mathbf{x}-\theta) \leq x_{\alpha; m} \right\}, 
\end{align}
it holds asymptotically that $\Prb(\hat{\theta}_n \in C_\alpha) = \alpha$. We obtain an estimate $\hat{C}_\alpha$ if we replace the asymptotic covariance $\Sigma$ by the sample covariance $\hat{\Sigma}$ and the true $\theta$ by its estimate $\hat{\theta}_n$. In particular cases the euclidean norm $\lVert \theta \rVert$ is actually the quantity of interest. From the definition of the confidence region $C_\alpha$, it is possible to directly derive an asymptotic confidence interval for the estimator $\lVert \hat{\theta}_n \rVert$, since asymptotically $\Prb(\lVert \hat{\theta}_n \rVert \in I_\alpha =[\max_{\mathbf{x} \in C_\alpha} \lVert \mathbf{x} \rVert, \min_{\mathbf{x} \in C_\alpha} \lVert \mathbf{x} \rVert]) \geq \alpha$. Using the estimated region $\hat{C}_\alpha$ we obtain an estimated confidence interval $\hat{I}_\alpha$. 

We are now applying the observations on Wald statistics to derive approximations to the uncertainty in the estimation of frequency domain causal quantities. In the following we denote by $\hat{\acov}^\Lag(\bullet, \bullet)$ the estimated asymptotic covariance obtained by substituting the OLS estimate $\hat{\Phi}^\Lag$ into the rational expression for $\acov^\Lag(\bullet, \bullet)$. Let us return to the problem of testing whether the process $v$ drives the process $w$ along a set of paths $\Pi$. If the asymptotic covairance $\acov^\Lag(\Pi, \Pi)$ is positive definite, then there are three approaches to this question, and each approach has a corresponding Wald test. First, we could test for each path $\pi\in \Pi$ whether the estimated path function $\fsep^{(\pi)}$ is significantly non-zero. By doing this, we simultaneously test whether all the link functions of the links on $\pi$ are significantly non-zero. Under the null hypothesis of no effect, we obtain the Wald statistic:
\begin{align}\label{eq: test statistic path function}
    T[\hat{\fsep}^{(\pi)}]^\top[\hat{\acov}^\Lag(\pi, \pi)]^{-1} [\hat{\fsep}^{(\pi)}] \to_d \chi^2(2)
\end{align}
Second, we could test whether there is a significant causal effect along at least one path $\pi \in \Pi$, which is equivalent to testing whether $\sum_{\in \Pi} |\hat{\fsep}^{(\pi)}(z)|$ is significantly non-zero. Under the null hypothesis of no effect along any path $\pi\in \Pi$ we get the statistic:
\begin{align}\label{eq: test statistic set of path functions}
    T[\hat{\fsep}^{(\pi)}]_{\pi \in \Pi}^\top[\hat{\acov}^\Lag(\Pi, \Pi)]^{-1} [\hat{\fsep}^{(\pi)}]_{\pi \in\Pi} \to_d \chi^2(2|\Pi|)
\end{align}
Finally, we could test whether the estimated total effect function $\hat{\fsep}^{\Pi} = \sum _{\pi \in \Pi} \hat{\fsep}^{(\pi)}$ is significantly different from zero. Under the null hypothesis we obtain the Wald statistic:
\begin{align}\label{eq: test statistic total effect}
    T[\hat{\fsep}^{\Pi}]^\top[\hat{\acov}^\Lag(\hat{\fsep}^{\Pi}, \hat{\fsep}^{\Pi})]^{-1} [\hat{\fsep}^{\Pi}] \to_d \chi^2(2)
\end{align}

The estimation of the function $\fsep^{\Pi}$ is subject to uncertainty, and this uncertainty can be approximated using the estimated asymptotic covariance matrix $\hat{\acov}^\Lag(\hat{\fsep}^\Pi, \hat{\fsep}^\Pi)$. For the significance level $\alpha$,
the construction (\ref{eq: asymptotic confidence region}) yields a region $\hat{C}_{\alpha}(z) \subset \R^2$ such that asymptotically $\Prb(\hat{\fsep}^{\Pi}(z) \in \hat{C}_\alpha(z)) = \alpha$. For interpretation and visualisation, however, it is more convenient to consider the absolute value $|\fsep^{\Pi}|$, that is, the strength with which the process $v$ influences the process $w$ along $\Pi$. By the construction presented at the beginning of this section we get an estimated confidence interval $\hat{I}_\alpha(z)$ for the estimated effect strength $|\hat{\fsep}^\Pi(z)|$. 

Analogously, one can approximate the uncertainty in the estimated spectral contribution of each ancestor $u\in \anc(v)$ through $v$ to $w$ along $\Pi$. Regarding the uncertainty in the estimate of the total contribution $\sd^\Pi$, we recall that it is equal to the sum over the squared absolute values of the individual contributions $\hat{\csep}^{\Pi_u}$. So the total contribution is non-zero if and only if $\csep^{\Pi_u}$ is significantly non-zero for at least one ancestor $u \in \anc(v)$. Thus, whether the process $v$ contributes significantly to the dynamics of $w$ can be decided with the following Wald statistic: 
\begin{align*}
    T[\hat{\csep}^{(\Pi_u)}]_{u \in \anc(v)}^\top [\hat{\acov}^\Lag(\hat{\csep}^{\Pi_u}, \hat{\csep}^{\Pi_{u'}})]_{u,u'\in \anc(v)}^{-1} [\hat{\csep}^{(\Pi_{u'})}]_{u' \in \anc(v)} \to_d \chi^2(2 |\anc(v)|)
\end{align*}
Finally, to quantify the uncertainty in the estimated total spectral contribution $\hat{\sd}^\Pi$ at a given significance level $\alpha$, we denote by $\hat{C}_\alpha(z) \subset \R^{2|\anc(v)|}$ the elliptical confidence region for the multivariate estimator $[\hat{\csep}^{\Pi_u}(z)]_{u \in \anc(v)}$. Since $\sd^{\Pi} = \lVert [\hat{\csep}^{\Pi_u}(z)]_{u \in \anc(v)} \rVert^2$ it follows that $\{ x^2 \mid x \in \hat{I}_\alpha(z)\}$ is an approximate confidence interval for the total contribution estimator.
\begin{remark}\label{rem: robust wald test}
    Path coefficients and spectral contributions admit a clear interpretation with respect to the stochastic process. However, they depend non-linearly on the coefficients of the SVAR parameter and possibly involve a large number of coefficients. These two factors may yield highly uncertain estimators. As a result one may need fairly long time series in order to statistically recognize a given effect by one of the tests presented in this subsection. If the objective is only to decide on the existence of a causal effect along a set of paths at a some frequency $z \in S^1$, then one could make use of the observation
    \begin{align*}
        \fsep_{v,w}(z) = \frac{\varphi_{v,w}(z)}{1-\varphi_{w,w}(z)} =0 \iff \varphi_{v,w}(z) = 0,
    \end{align*}
    where the right hand side in contrast to the left hand side depends linearly on the SVAR coefficients. In particular, $\fsep^{(\pi)} = 0 $ iff $ \varphi^{(\pi)} = 0$. The asymptotic covariance among the estimators $(\hat{\varphi}_{v,w}(z))_{v\in \pa(w)}$ can be obtained from Proposition \ref{prop: asymptotic covariance link functions} by setting the matrix $A(z)= \mathbb{I}_2$ and $A_{v}(z)= \mathbf{0}$ for $v \in \pa(w)$ as in Proposition \ref{prop: asymptotic covariance link functions}. The asymptotic covariance between estimators for (weighted) path functions can then be computed by the formulas as presented in Proposition \ref{prop: asymptotic covariance path function} and \ref{prop: asymptotic covariance weighted path function}, which in turn can then be used to formulate Wald type tests as in (\ref{eq: test statistic path function}) and (\ref{eq: test statistic set of path functions}).
\end{remark}
\begin{figure*}
    \centering
    \includegraphics[width=0.9\textwidth]{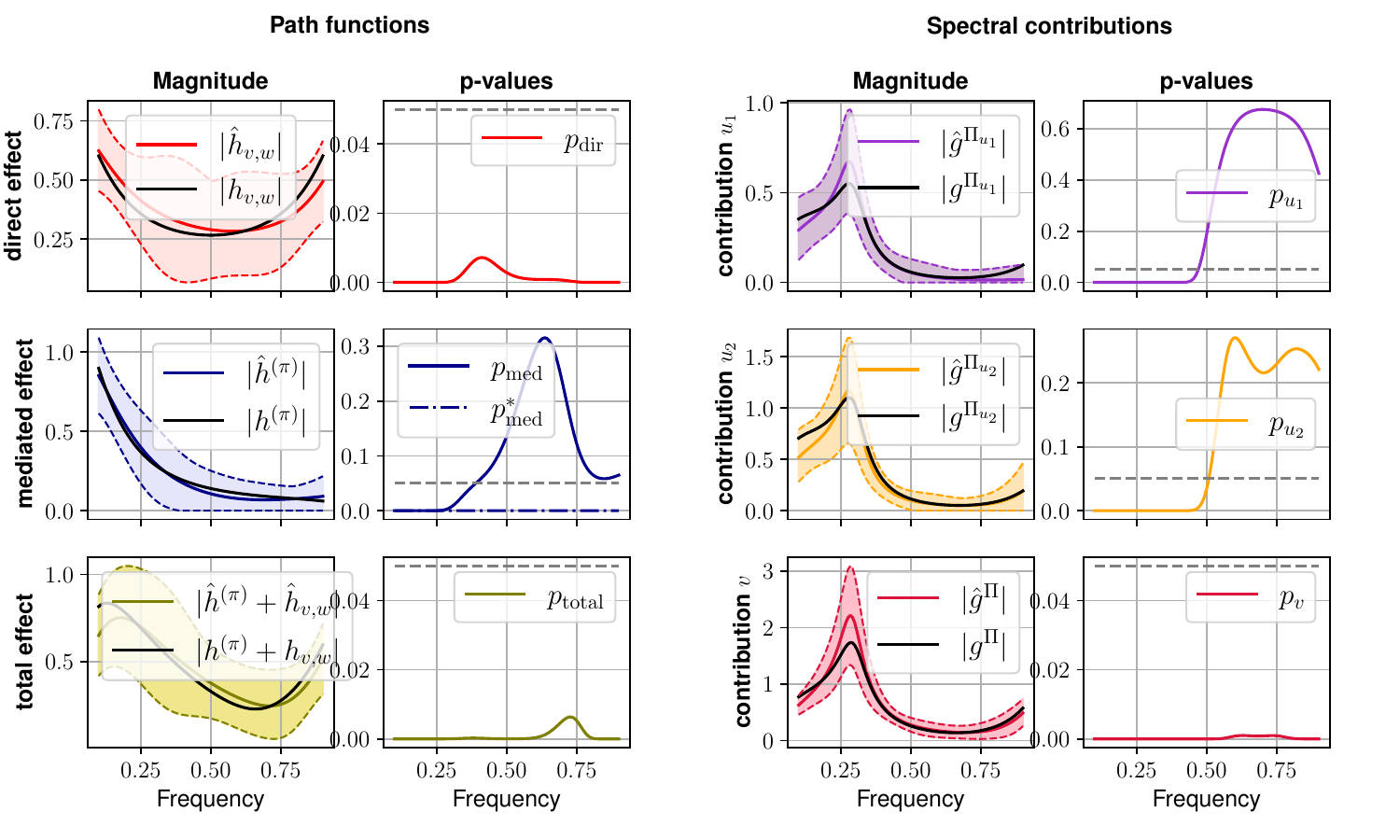}
    \caption{Statistical analysis of the causal structure of the process shown in Figure \ref{fig:process graph} based on a synthetically generated time series of length $500$. Analysis of the estimates for direct mediated and total effects of $v$ on $w$ are shown on the left. The analysis of the spectral contribution of $\anc(v)$ on $w$ via $\Pi = \paths(v,w)$ is displayed on the right.}
    \label{fig: example effect and contribution}
\end{figure*}
\begin{example}
To illustrate the frequency domain Wald tests and uncertainty approximations we generate a synthetic time series of length $500$ for the process graph $G$ displayed in Figure \ref{fig:process graph} and underlying time series graph $\G$ and SVAR coefficient vector $\Phi$ (see supplementary material). With this time series we analyse and test how the process $v$ is driving $w$. To test whether the data supports a significant causal effect of $v$ on $w$ we employ either of the tests (\ref{eq: test statistic path function})-(\ref{eq: test statistic total effect}). In the left part of Figure \ref{fig: example effect and contribution} we display the estimated magnitudes of the direct, mediated and total effect function together with the estimated $95\%$ confidence region and $p$-values. Additionally, we compute the true magnitudes of the respective effects and note that for each effect the true magnitudes lie within the estimated confidence region. While the direct and total effect are detected for all frequencies with $95\%$ confidence, the mediated effect does not get recognised for frequencies greater than approximately $0.4 \pi$. However, we detect the effect with the Wald test described in Remark \ref{rem: robust wald test}, as indicated by the p-values $p^\ast_\mathrm{med}$ (dash dotted line) in the second row and second column of the left part in Figure \ref{fig: example effect and contribution}. The right part of Figure \ref{fig: example effect and contribution} displays the estimated spectral contributions of each of the processes $u_1, u_2,v$ to the process $w$, together with the estimated $95\%$ confidence regions and the true contributions. Also here we observe that the true values lie within the estimated confidence region. In addition we display the p-values of the estimated contributions.
\end{example}

\section{Asymptotic efficiency}\label{sec: asymptotic efficiency}
Suppose we seek to estimate a given effect or spectral contribution from time series data. We saw that any such quantity is a rational function in the SVAR parameters. In the previous section we considered one type of estimator, i.e. estimating the SVAR coefficient and feeding it into the rational function. In general, however, there might be many more possible estimators with which the desired quantity can be estimated. The objective of this section is to identify the estimator considered in the previous section as the asymptotically optimal, i.e. the estimator with lowest asymptotic variance, with respect to a specific class of estimators. With this we generalise the main result in \citep{Guo2020EfficientLS} to causal effects in the frequency domain. For the rest of this section, we fix a time series DAG $\G$ having order $p$. We assume that the associated contemporaneous graph $\G_0$ and the possibly cyclic process graph $G$ are known. Once more, we encode the information on $\G$ by a collection of time lagged relations $\Lag = \{\Lag_v \}_{v\in V}$ such that (\ref{condition: time lags}) holds for every $v$. 

In what follows we assume that the processes $V = \{ v_1, \dots, v_m \}$ are topologically ordered with respect to the contemporaneous graph $\G_0$, that is $v < w$ only if $w \not\in \anc_0(v)$, where $\anc_0(v)$ are the ancestors of $v$ with respect to $\G_0$. 
In this section, we view any stable SVAR parameter vector as a tuple of matrices, that is $\Phi = [\Phi(k)]_{k =0}^q  \in \R^{m \times m} \otimes \R^q$, where $q \geq p$. So the entry $\Phi_{j,i}(k) = \phi_{v_i, v_j}(k)$ if $v_i \to_k v_j$ and $0$ otherwise. Due to the ordering of the process indices $V$ it follows that $\Phi(0) \in LT(m)$, which is the subspace of lower triangular $m \times m$-dimensional matrices. For the arguments to come it will be important to note that there is an explicit relation between the ACS $\Sigma$ and the parameter tuple $(\Phi, \Omega)$. This relation is based on the following SVAR recursion: 
\begin{align*}
    \X(t) &= \sum_{k= 0}^q \Phi(k)\X(t-k) + \eta(t) \\
    &=\sum_{k=1}^q \tilde{\Phi}(k)\X(t-k) + \tilde{\eta}(t) & \tilde{\Phi}(k) &= (\mathbb{I}- \Phi(0))^{-1} &\tilde{\eta}= (\mathbb{I}- \Phi(0))^{-1} \eta(t)
\end{align*}
This representation reveals the recursive structure (Yule-Walker equations \cite{brockwell2009time}) in the ACS sequence:
\begin{align} \label{eq: acs recursion}
    \Sigma(k) &= 
    \begin{cases}
        \tilde\Phi(1)\Sigma(-1) + \cdots + \tilde\Phi(p)\Sigma(-p) + \tilde\Omega &, k =0 \\
        \tilde\Phi(1)\Sigma(k-1) + \dots + \tilde\Phi(p)\Sigma(k-p) &, k \neq 0
    \end{cases}
\end{align}
Since $\Sigma(k) = \Sigma(-k)^\top$ it then follows that the ACS $\Sigma$ is uniquely determined by the first $p$ entries $\Sigma(0), \dots, \Sigma(q-1)$. Suppose $\Sigma = (\Sigma(k))_{k \in \Z}$ is a collection of $m \times m $-dimensional matrices and $I, J \subset \Z$ finite subsets, then we consider the following block matrix constructions: 
\begin{align} \label{eq: symmetric pos def}
    \Sigma(I) &= \begin{bmatrix}
        \Sigma(i)
    \end{bmatrix}_{i \in I} \in \R^{m \times m } \otimes \R^{|I|} & \Sigma(I, J)&=\begin{bmatrix}
        \Sigma(j-i)
    \end{bmatrix}_{i \in I, j \in J} \in \R^{m\times m } \otimes \R^{|I| \times |J|}
\end{align} 
Based on these matrix constructions we consider the following space of matrix valued tuples: 
\begin{align*}
    \mathcal{E}_q(m) = \{ (\Sigma(i))_{-q \leq i \leq q} \mid \Sigma(k) = \Sigma(-k)^\top, \Sigma([1,q], [1,q]) \in \mathrm{PD}(m(p-1)) \} 
\end{align*}
The following Proposition establishes a diffeomorphic relation between the parameter pair $(\Phi, \Omega)$ and the ACS $\Sigma$ of any stable SVAR process. 
\begin{proposition}\label{prop: diffeomorphism}
    There are open sets $U \subset LT(m) \times (\R^{m \times m} \otimes \R^q) \times \mathrm{diag}_{+}(m)$ such that every parameter pair $(\Phi, \Omega) \in U$ defines a stable process, and an open set $V \subset \mathcal{E}_q(m)$ such that there is a  diffeomorphism $U \cong V$. 
\end{proposition}
\begin{proof}
    Let us begin with the map in the inverse direction. So we pick some $(\Sigma(i))_{-p \leq i \leq p} \in \mathcal{E}_p(m)$. Then we set
    \begin{align*}
        \tilde\Phi &= \Sigma([1,q], [1,q])^{-1} \Sigma([1,p]) & \tilde\Omega = \Sigma(0)- \sum_{j=1} \Sigma(j)\tilde\Phi(j),
    \end{align*} 
    which is the unique solution to the recursion in Equation (\ref{eq: acs recursion}). 
    Then we obtain $\Phi(0)$ and $\Omega$ from $\tilde\Omega$ by recursive regression along the topological order of $\G_0$, and for $k > 1$ we set $\Phi(k) = (\mathbb{I}- \Phi(0))^{-\top}\tilde\Phi(k)$.
    
    Now let us consider a collection of coefficients $\Phi \in  UT(m) \in \times (\R^{m \times m} \otimes \R^q)$ and positive diagonal matrix $\Omega \in \mathrm{diag}_{+}(m)$. 
    The ACS $\Sigma$ of the SVAR process specified by $(\Phi, \Omega)$ then satisfies the recursion $(\ref{eq: acs recursion})$. Consequently, we obtain for any two $v,w \in V$ and $-p+1 \leq j \leq p-1$ the relation
    \begin{align*}
        \sigma_{v,w}(j) &= \tilde\omega_{v,w}(j) + \sum_{x,y \in V} \sum_{k_x=1}^p \sum_{k_y=j+1}^p \tilde\phi_{x,v}(k_x) \sigma_{x,y}(k_x - k_y) \tilde\phi_{y,w}(k_y-j), \\
        \tilde\omega_{v,w}(j) &= \begin{cases}
            [\tilde\Omega]_{v,w} & \text{if } j = 0 \\
            0 & \text{if } j \neq 0 
        \end{cases}
    \end{align*}
    In addition, it must hold that $\sigma_{v,w}(j) = \sigma_{w,v}(-j)$. That means for every $\Phi$ we get a linear system of equations that can be represented by a quadratic $m(q-1) \times m (q-1)$-dimensional matrix. We observe that the determinant of this matrix is a non-zero polynomial in the parameters $\Phi$. This follows as the coefficient of the zero degree monomial is one. So we conclude that the set of parameters $\Phi$ for which the corresponding system of equations has a unique solution is open. For such parameters, the two maps constructed in this proof are rational functions that are inverse to each other. Furthermore, the set of SVAR parameters that defining a stable process contains a set that is open in $LT(m) \in \times (\R^{m \times m} \otimes \R^q)$, see \cite{reiter2023formalising}. This finishes the proof.  
\end{proof}

We now define the class of estimators among which we seek to identify the asymptotically optimal one. First, let $U\subset V$ be a set of processes, and $w\in V \setminus U$ a target process. For each $u \in U$ we furthermore pick a possibly infinite subset of paths $\Pi_u\subset \paths(u,w)$. Then, for some $z \in S^1$ we denote by $\tau(z)$ either the vector of path functions $[\fsep^{\Pi_u}(z)]_{u \in U}$ or of weighted path functions $[\csep^{\Pi_u}]_{u \in U}$.
We require the sets of paths $\Pi_u$ to be such that they correspond to a \emph{(controlled) causal effect} \cite{reiter2023formalising} of $U$ on $w$. This requirement ensures that $\tau$ is a rational function in the SVAR parameters.

Similarly to \citep{Guo2020EfficientLS} we introduce the space of all consistent estimators for $\tau(z)$ that can be written as differentiable functions of a finite collection of the (estimated) ACS entries of the SVAR process, i.e. for $q \geq p $ we consider
\begin{align*}
    \mathcal{T}^{(q)}(z) = \left\{ \hat{\tau}(z): \mathcal{E}_q(m) \to \R^2 \otimes \R^{|U|} \mid \text{$\hat{\tau}$ is differentiable and a consistent estimator of $\tau$} \right\}
\end{align*}
Since $\tau(z)$ is a rational function in the SVAR parameters it follows from Proposition \ref{prop: diffeomorphism} that $\mathcal{T}$ is non-empty. 
The set $\mathcal{T}^{(q)}(z)$ contains all possible estimators that are based on least square regression. In particular, this includes the estimator $\hat{\tau}^\Lag$ with which we refer to the estimator that evaluates the rational function defining $\tau$ on the OLS estimate $\hat{\Phi}^\Lag$. 
But this is not the only type of estimator for frequency domain causal quantities. The following example indicates a less obvious estimator for an indirect causal effect in the frequency domain. 
\begin{example}
    We consider an SVAR process with process graph $G = (V,D)$, where $V = [3]$ and $D=\{ 1 \to 2, 2 \to 3 \}$. The underlying time series graph $\G$ has order $1$. The lagged parent sets are $L_{i}^\G = \pa(i) \times [0,1] \cup \{i\} \times \{1\}$ for every $i \in V$. The total causal effect of $1$ on $3$, which is the the path function of $\pi = 1 \to 2 \to 3$, i.e. 
    \begin{align*}
        \fsep^{(\pi)}(z) &= \frac{p_0 + p_1 z  + p_2z^2}{1- q_1 z - q_2 z^2}, \\
        p_0 &= \phi_{1,2}(0)\phi_{2,3}(0) & p_1 &= \phi_{1,2}(0)\phi_{2,3}(1) + \phi_{1,2}(1)\phi_{2,3}(0) & p_2 &= \phi_{1,2}(1)\phi_{2,3}(1) \\
        & & q_1& = \phi_{2,2}(1) + \phi_{3,3}(1) & q_2 &=\phi_{2,2}(1)\phi_{3,3}(1)
    \end{align*}
    Then we denote by $\lambda^{(\pi)}$ the path filter of $\pi$, as defined and characterised as sums of path coefficients on the time series graph $\G$ in \citep{reiter2023formalising}. With this filter we can establish a polynomial equation system for the coefficients $p_k$ and $q_l$ as follows 
    \begin{align*}
        \lambda^{(\pi)}(0) &= p_0 \\
        \lambda^{(\pi)}(1) &= q_1 p_0 + p_1 \\
        \lambda^{(\pi)}(2) &= (q_1^2 - q_2) p_0 +  q_1 p_1 + p_2 \\
        \lambda^{(\pi)}(3) &= (q_1^3 - 2 q_2 q_1)p_0 + (q_1^2 - q_2) p_1 +  q_1 p_2 \\
        \lambda^{(\pi)}(4) &= (q_1^4 - 3q_2q_1^2 + q_2^2)p_0 + (q_1^3 - 2 q_2 q_1)p_1 + (q_1^2 - q_2) p_2
    \end{align*}
    The values of the path filter $\lambda^{(\pi)}$ can be obtained with OLS regression from the ACS information $\Sigma_{1,1}(i)$ for $i=[1,5]$ and $\Sigma_{3,1}(i)$ for $i\in [0,4]$. The triangular structure of the first three equations lets us rewrite the above system as a system of two equations for the parameters $q_1, q_2$. The solution to this system of equations is unique if numerator and denominator of $\fsep^{(\pi)}$ are coprime, which is the case for generic SVAR parameter choices \cite{reiter2024causal}. This yields a consistent estimator $\hat{\tau} = \hat{\fsep}^{(\pi)} \in \mathcal{T}_a^{(5)}$. 
\end{example}
We are now prepared to state the main result of this section, which can be seen as a frequency version of \citep{Guo2020EfficientLS}.
\begin{theorem}\label{thm: asymptotic efficiency}
    Let $G=(V,D)$ be a process graph with underlying time series graph $\G$ with order $p$, and let $\tau$ be a vector of sums of (weighted) path functions. If the OLS estimators $\hat{\tau}^{\Lag_\G}(z)$ is asymptotically normal, then for every $q \geq p$ and every asymptotically normal estimator $\hat{\tau}(z) \in \mathcal{T}^{(q)}(z)$ it holds that 
    \begin{align}\label{eq: absolute asymptotic efficiency}
          \acov(\hat{\tau}(z)) \succeq \acov(\hat{\tau}^{\Lag^\G}(z)),
    \end{align}
    where for two positive semi-definite matrices $A,B$ we write $A \succeq B$ if $A - B$ is positive semi-definite. Furthermore, suppose $\Lag'$ and $\Lag''$ are collections of time lagged relations such that both satisfy condition (\ref{condition: time lags}) and $\Lag'_v \subset \Lag''_v$ for every $v \in V$, then 
    \begin{align}\label{eq: relative asymptotic efficiency}
        \acov(\hat{\tau}^{\Lag''}(z)) \succeq \acov(\hat{\tau}^{\Lag'}(z)).
    \end{align}
\end{theorem}

\section{Impact of the solar cycle on the NAO}\label{sec: application}
\subsection{Background}
In this section, we employ the tests for causal effects in the frequency domain from section \ref{subsec: asymptotic uncertainty and significance} to analyse the possible effect of variations in solar activity, represented by sun spot numbers, on northern European climate during winter time, represented by the North Atlantic Oscillation (NAO).
Solar activity undergoes an approximately 11-year cycle, where slightly enhanced solar insulation coincides with increases in the number of sunspots. While the direct impact on global climate of the solar cycle is very small consistent with the energy budget of the planet \citep{gray2013responsesolarcycle, gray2010solarinfluenceclimate}, it has been suggested that the sunspot cycle can influence the NAO through its effect on temperature gradients through the atmosphere \cite{Lockwood2012SolarIO}. The NAO represents the tendency towards stronger or weaker westerly circulation. Empirical analysis suggests that during minima in the sunspot cycle there is a tendency for the NAO to become negative, resulting in northerly and easterly flow, which is conducive to cold winter weather, as observed during the 2009/2010 solar minimum. However, the atmospheric circulation can also be influenced by remote effects of the El Nino Southern Oscillation (ENSO), the Atlantic multidecadal oscillation in sea surface temperatures (AMO), aerosol effects from human influences and volcanic forcing are reflected through their optical depth (AOD). 
\begin{figure}
    \centering
    \includegraphics[width=0.5\textwidth]{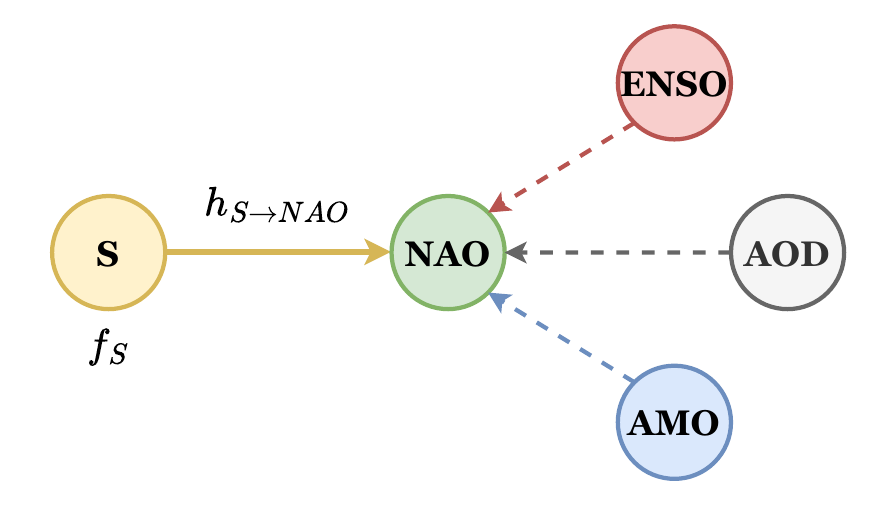}
    \caption{The process graph based on which the solar effect on NAO is investigated, where S stands for solar activity represented by sun spot numbers.}
    \label{fig:application process graph}
\end{figure}
\subsection{Data}
We construct the time series for our analysis from the following publicly available sources: NAO \citep{jonesNAO1997}, ENSO \citep{raynerSST2003}, AMO \citep{enfieldAMO2001}, AOD \citep{satoAOD1993}, SNS \citep{sidc}. From each of those monthly time-series we construct yearly winter series by averaging the December, January and February values. We normalise each time series by deducting its mean and dividing by its variance. The Diceky-Fuller test \citep{fuller1979unitroottest} applied to the annual winter time time series suggests that the SNS and AMO time series are non-stationary for which we adjust by considering the time series of first order differences \citep{lutkepohl2005new} instead. These pre-processing steps give five annual time series that are normalised and stationary, as suggested by the Dickey Fuller test with $95 \%$ confidence. Those time series cover the period from 1870 until 2010. 
\subsection{Result}
The hypothesised causal structure between the processes is represented by the process graph shown in Figure \ref{fig:application process graph}. The short length of the time series makes it necessary to constrain the underlying time series graph to reduce the number of parameters to be estimated. We follow the choices explained in \citep{gray2013responsesolarcycle}.
The magnitude of the possible direct influence of solar activity on the NAO is the function $|\fsep_{\SNS, \NAO}|$ and the possible spectral contribution of the solar process to the NAO is $\csep_{\SNS \to \NAO} = |\fsep_{\SNS \to \NAO}|\sqrt{\sd_{\SNS}}$, which is the fraction of the spectral density of NAO anomalies determined by solar activity anomalies. Solar activity anomalies oscillate with a period between 10 and 11 years. We therefore restrict our analysis to oscillations with periods between 8 and 13 years. We use our assumptions about the structure of the time series graph to compute the OLS estimator $\hat{\Phi}$ for the SVAR coefficients and plug it into the formulas for $\fsep_{\SNS\to \NAO}$ and $\csep_{\SNS\to \NAO}$ respectively. In Figure \ref{fig: sns nao direct effect} we show their estimated magnitudes together with their $95 \%$ confidence interval and their p-values. The Wald test detects a significant causal effect and spectral distribution for all frequencies considered with at least $95\%$ confidence, as can be seen from the plotted p-values.
\begin{figure*}
    \centering
    \includegraphics[width=0.9\textwidth]{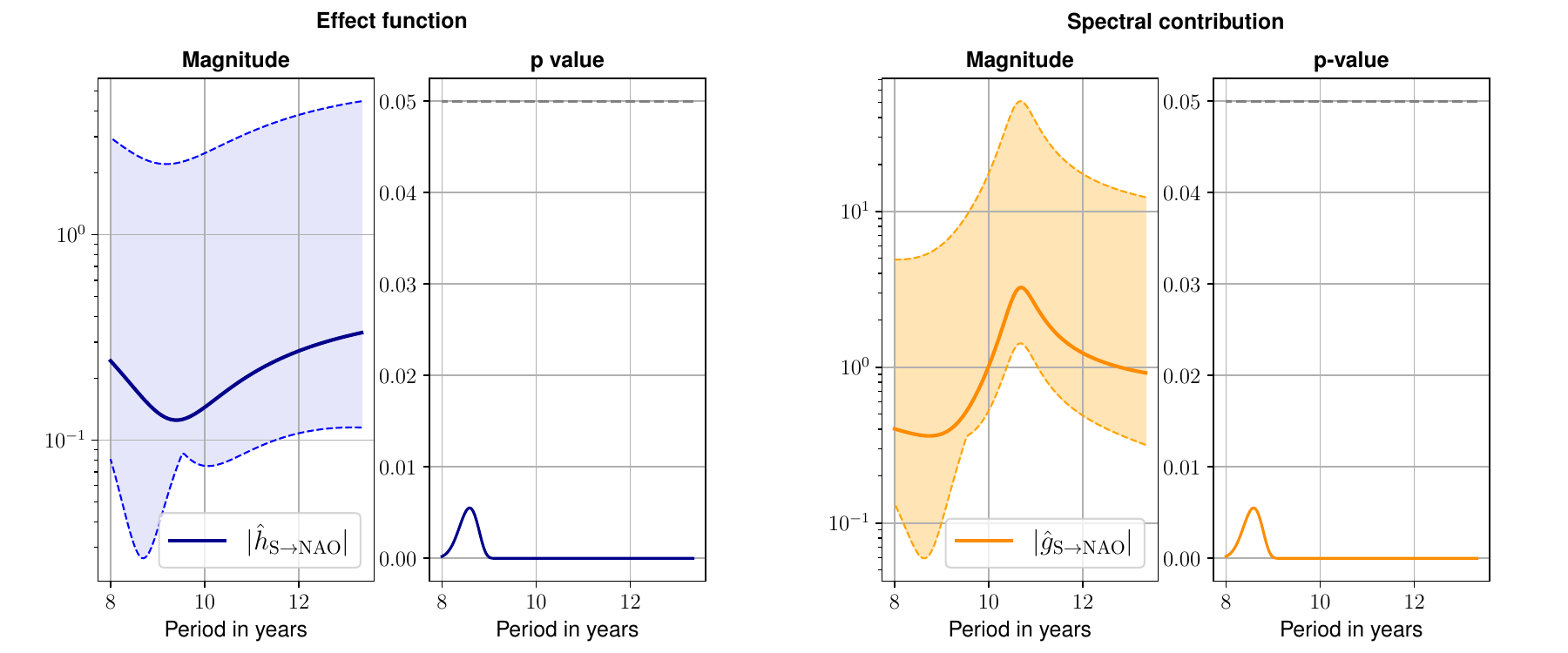}
    \caption{The statistical analysis of the solar influence on the NAO for oscillations with periods between 8 and 13 years.}
    \label{fig: sns nao direct effect}
\end{figure*}
\section{Discussion}
Frequency domain path functions and spectral contributions compactly capture and visualise the causal structure of linear SVAR processes. As they can be estimated from observations, they could be useful for investigating many scientific questions. However, in order to draw conclusions from an estimate, it is important to assess its uncertainty. For large sample sizes, this uncertainty can be approximated by the asymptotic Gaussian distribution of the estimator, for which we computed structured expressions using the delta method. Based on these expressions, we then constructed Wald type tests to assess the significance of the estimated frequency domain causal quantity. 
The problem with approximating distributions by the normal distribution obtained by the delta method is that it may be far from the true distribution, for instance if the sample size is too small. How many samples you need for the approximated distribution to be meaningful depends, for example, on the number of SVAR parameters involved, or how far the true SVAR parameter is from a parameter that defines an unstable process. In addition, problems may arise if the SVAR coefficient is close to a parameter for which the rank of the gradient used to compute the asymptotic distribution degenerates. We hope that this work will lead to further research into the challenges and uncertainties of estimating causal quantities in the frequency domain.

Furthermore, for any given causal quantity in the frequency domain, we could identify the estimator with the lowest asymptotic variance among all estimators that can be expressed as a differentiable function in finitely many entries of the autocovariance. This optimal estimator is constructed in two steps: first, one computes all the SVAR coefficients by regressing on the lagged parents as prescribed by the time series graph. In the second step, we evaluate the rational function (in the SVAR parameters) that describes the causal quantity of interest on the estimated SVAR coefficients. However, the construction of this estimator requires knowledge of the complete time series graph, and this information may not always be available. Instead, it may only be possible to know the structure of the process graph together with an upper bound on the order of the SVAR process. In future work, it would be interesting to search for the asymptotically optimal estimator relative to some fixed information about the time series graph. Another direction in which this work could be extended is to consider the scenario where some of the processes are hidden, i.e. they influence the observable processes but cannot be measured directly. In a companion work \cite{reiter2024causal} we establish the theoretical foundations for identifying frequency domain causal effects in the presence of latent confounding processes. 

We believe that further research into the estimation of causal effects in the frequency domain could be not only theoretically interesting but also practically relevant, which we demonstrate with a case study of the influence of the solar cycle on the variability of the NAO. While previous studies have looked at lagged relationships between solar activity and the NAO, in this study we have analysed the effect directly at the relevant time scale in the frequency domain, where the solar variability pattern is clearly captured. Our results confirm a significant influence of the solar cycle on the NAO, which is consistent with previous studies on this question \citep{gray2010solarinfluenceclimate, gray2016solarcycleNAO}. However, our results also suggest that further investigation of this frequency domain effect is needed. In particular, the plots in Figure \ref{fig: sns nao direct effect} suggest that the approximate distribution of the estimator for the solar influence on the NAO may be degenerate. Thus, to argue convincingly for a significant causal effect of the solar process on the NAO on the considered time scale, extended investigations are most likely required.

With a view to further applications of this type, it would be attractive to extend the framework of this analysis so that frequency-wise causal effects could also be analysed when the underlying data is structured not only in time but also in space. 

\section*{Acknowledgement}
J.W. and J.R. received funding from the European Research Council (ERC) Starting Grant CausalEarth under the European Union’s Horizon 2020 research and innovation program (Grant Agreement No. 948112).
N.R., G.H. and J.R. received funding from the European Union’s Horizon 2020 research and innovation programme under Marie Skłodowska-Curie grant agreement No 860100 (IMIRACLI).
N.R. and G.H. thank Mike Evans and Carla Rösch for their help and encouraging discussions.


\appendix
\section{Technical preliminaries}
\subsection{Schur complements}
Our indexing conventions for vectors and matrices are as follows: Suppose $I,J$ are finite sets, then $\R^I$ denotes an $|I|$-dimensional real vector space whose dimensions are indexed by $I$. A matrix $A = (a_{i,j})_{i \in I, j \in J}\in \R^{I \times J}$ represents a linear map from $\R^{J}$ to $\R^{I}$. For subsets $I' \subset I $ and $J' \subset J$ we define $[A]_{I', J'} = (a_{i,j})_{i \in I', j \in J'} \in \R^{I' \times J'}$ as the submatrix of $A$ corresponding to the rows $I'$ and the columns $J'$. For any two integers $p \leq q$ we use $[p,q]$ to denote the set $\{i: p \leq i \leq q\}$, and we use $[p]$ as a shorthand for the set $[0,p]$.

Suppose $I$ and $J$ are finite disjoint sets, and $M \in \R^{I \cup J \times I \cup J}$ a block structured matrix
\begin{align*}
    M = \begin{bmatrix}
        A & B  \\
        C & D
    \end{bmatrix}
\end{align*}
The \emph{Schur complement} \cite{gallier2011geometric} of $M$ with respect to $I$ resp. with respect to $J$ are defined as 
\begin{align*}
    M_{J \cdot I} &\coloneqq D - C A^{-1} B & M_{I \cdot J} &\coloneqq A - B D^{-1} C.
\end{align*}
\begin{theorem}[Invertibility Schur complement \cite{gallier2011geometric}]
    Let $M \in \R^{I \cup J \times I \cup J}$ be a block structured matrix and suppose one of the  matrices $M, M_{I \cdot J}, M_{J \cdot I}$ is invertible, then also the other two matrices are invertible. In this case, the inverses of these matrices can be computed as follows 
    \begin{align*}
        M^{-1} &= \begin{bmatrix}
            A^{-1} + A^{-1}B(M_{J \cdot I})^{-1} CA^{-1} & -A^{-1}B(M_{J \cdot I})^{-1} \\
            (M_{J \cdot I})^{-1} CA^{-1} & (M_{J \cdot I})^{-1}
        \end{bmatrix} \\
        &= \begin{bmatrix}
            (M_{I \cdot J})^{-1} & -(M_{I \cdot J})^{-1}BD^{-1} \\
            -D^{-1}C(M_{I \cdot J})^{-1} & D^{-1} + D^{-1}C (M_{I \cdot J})^{-1} BD^{-1}
        \end{bmatrix}
    \end{align*}
\end{theorem}

\begin{corollary}\label{cor: schur inverison}
    Suppose the block structured  matrix $M \in \R^{I \cup J \times I \cup J}$ as above is invertible, then it holds that 
    \begin{align*}
        D^{-1} = [M^{-1}]_{J \cdot I},
    \end{align*}
    where the right hand side is the Schur complement of $M^{-1}$ with respect to $I$. 
\end{corollary}

\subsection{Time lagged regression}
Suppose $\mathcal{G} = (V \times \Z, \D)$ is a time series DAG of order $p$. The associated contemporaneous graph is denoted $\G_0= (V, \D_0)$ and the process graph $G= (V, D)$. We assume that the process indices $V$
are topologically ordered with respect to $\G_0$, i.e. for two distinct processes $v,w$ the relation $v < w$ implies $w \not \in \anc_0(v)$.
In the following, we assume that the SVAR parameter $\Phi$ defines a stable process $\X$. This assumption ensures that $\X$ is stationary, i.e. the mean sequence and ACS are time independent. One convenient way to ensure stability is to require that
\begin{align}\label{assimption: stability}
    \sum_{v \in V} \sum_{(u,k) \in \Lag_v} |\phi_{v,w}(k)| < 1.
\end{align}
By requiring this condition we get that the space of all considered SVAR parameters is an open semi-algebraic subset of $\R^{\G} = \prod_{v \in V} \R^{\Lag_v}$, i.e. the euclidean space whose dimensions are indexed by the lagged relations. 

Furthermore, let $q\geq p$ a non-negative integer, then we define $\Bar{\Lag} = \{ \Lag_v \}_{v \in V}$ as the collection of time lagged relations where 
\begin{align}\label{eq: time lags coarse}
    \tilde\Lag_v &= \{ u \in V \mid u < v \}  \times [0,q] \cup \{ u \in V \mid u \geq v \} \times [1,q].
\end{align}
For two collections of time lagged relations $\Lag'' = \{\Lag_v'' \}_{v \in V}$ and $\Lag' = \{\Lag'_v \}_{v \in V}$, we write $\Lag' \subset \Lag''$ if $\Lag'_v \subset \Lag''_v$ for every $v \in V$, and we define $\Lag'' \setminus  \Lag' \coloneqq \{\Lag''_v \setminus \Lag'_v \}_{v \in V}$. Throughout the supplementary material we will only consider collections of time lagged relations $\Lag$ that either satsify condition (\ref{condition: time lags}) from the main paper or satisfy the coarser condition 
\begin{align} \label{condition: time lags coarse}
    \Lag^\G \subset \Lag \subset \Bar{\Lag}.
\end{align}
Condition (\ref{condition: time lags coarse}) guarantees that the OLS estimator $\hat{\Phi}^{\Lag}$, as defined in the main paper, is a consistent estimator for $\Phi$. 

In the main part of the paper we defined the random vector $\X^{\Lag_v}(t)$ containing the super set of lagged parents of $\X_v(t)$. For the purpose of the computations we will structure $\X^{\Lag_v}$ as follows 
\begin{align*}
    \X^{\Lag_v}(t) &= \begin{bmatrix}
        \X^{\Lag_v}_u(t)
    \end{bmatrix}_{v \in V}, & \X^{\Lag_v}_u(t)& = \begin{bmatrix}
        \X_u(t-k)
    \end{bmatrix}_{(u,k) \in \Lag_v}
\end{align*}
The covariance matrix associated with the process $\X^{\Lag_v}(t)$ is defined as follows:
\begin{align*}
    \Sigma^{\Lag_v} \coloneqq \E[\X^{\Lag_v}(t) (\X^{\Lag_v}(t))^\top] \in \R^{\Lag_v \times \Lag_v}
\end{align*}
We now equip this matrix with a block structure. For $u \in V$ we define $\Lag_v(u)$ to be the set $\{u \} \times [q] \cap \Lag_v$, and for $u_1, u_2 \in V$ we define the matrix
\begin{align*}
    \Sigma_{u_1,u_2}^{\Lag_v} &\coloneqq \E[(\X^{\Lag_v}_{u_1}(t)) (\X^{\Lag_v}_{u_2}(t))^\top] \\
    &= [\Sigma^{\Lag_v}]_{\Lag_v(u_1), \Lag_v(u_2)} \in \R^{\Lag_v(u_1) \times \Lag_v(u_2)}. 
\end{align*}
So let us structure the covariance matrix like this 
\begin{align*}
    \Sigma^{\Lag_v} &= \left[
        \Sigma^{\Lag_v}_{u_1,u_2}
    \right]_{u_1,u_2 \in V}
\end{align*}
Finally, we equip the associated precision matrix with the very same block structure, i.e.
\begin{equation}
    \begin{split}
        P^{\Lag_v} &= \omega_v \left(\Sigma^{\Lag_v}\right)^{-1} \\
        &=\left[P^{\Lag_v}_{u_1,u_2}
    \right]_{u_1, u_2 \in \pa(v) \cup \{v\}} \in \R^{\Lag_v \times \Lag_v},
    \end{split}
\end{equation}
where $P^{\Lag_v}_{u_1,u_2} = [P^{\Lag_v}]_{\Lag_{v}(u_1), \Lag_v(u_2)} \in \R^{\Lag_v(u_1), \Lag_v(u_2)}$.

For the proof of Proposition \ref{prop: asymptotic covariance path function}, Proposition \ref{prop: asymptotic covariance weighted path function} and Theorem \ref{thm: asymptotic efficiency} from the main paper, we need two auxiliary lemmas. The first is a time series adaptation of Lemma 23 in \citep{Guo2020EfficientLS}. The second establishes a block diagonal structure on the asymptotic covariance of the OLS estimator $\hat{\Phi}^\Lag$. 
\begin{lemma}[Time series adaption of \citep{Guo2020EfficientLS}]\label{lemma: asymptotic difference } 
    Let $v\in V$ a process and $\Lag_v$ be a set of time lagged relations as specified above, then it holds that 
    \begin{align*}
        \hat{\Phi}^{\Lag_v} - \Phi^{\Lag_v} &= \frac{1}{T} \sum_{t=1}^T(\Sigma^{\Lag_v})^{-1}\X^{\Lag_w}(t)\eta_v(t) + O_p(T^{-1}).
    \end{align*}
\end{lemma}
\begin{proof}
    The statement can be shown by the same computations as in the proof of Lemma 23 in \citep{Guo2020EfficientLS}. In order to conduct these computations we use the fact that the auto covariance estimator is asymptotically normal if the underlying process is stationary \citep{hanan1976serialcovariance}. 
\end{proof}
In the following, we assume that the vertices $V$ are topologically ordered with respect to the acyclic contemporaneous graph $\G_0 =(V, \D_0)$. That means, $V = \{ v_1 , \dots , v_m \}$ where $v_i < v_j $ implies $v_j \not \in \anc_0(v_i)$. 
\begin{lemma}\label{lemma: block diagonal structure of }
    Suppose $\Lag = \{ \Lag_v \}_{v \in V}$ is a collection of time lags such that $\Lag^\G \subset \Lag \subset \Bar{\Lag}$. If we denote by $(\hat{\Phi}^{L_v})_{v \in V}$ the OLS estimators for $\Phi^{\Lag_v}$ obtained by regressing $\X_v(t)$ on $\X^{\Lag_v}(t)$, then its asymptotic covariance has the following block diagonal structure 
    \begin{align}
        \sqrt{n}\begin{bmatrix}
            \hat{\Phi}^{\Lag_1} - \Phi^{\Lag_1}\\
            \vdots  \\
            \hat{\Phi}^{\Lag_m} - \Phi^{\Lag_m}
        \end{bmatrix} \to_d \mathcal{N}(\mathbf{0}, \mathrm{diag}[\omega_2 (\Sigma^{\Lag_1})^{-1}, \dots , \omega_m  (\Sigma^{\Lag_m})^{-1}])
    \end{align}
\end{lemma}
\begin{proof}
    Let $i, j \in [1, m]$ and suppose $i < j$. We wish to show that 
    \begin{align*}
        \lim_{T \to \infty} T \E[(\hat{\Phi}^{\Lag_i}- \Phi^{\Lag_i})(\hat{\Phi}^{\Lag_j} - \Phi^{\Lag_j})^\top] = 0,
    \end{align*}
    where $\hat{\Phi}^{\Lag_i}$ is the OLS estimator based on a sample of length $T$. Using Lemma \ref{lemma: asymptotic difference } we get 
    \begin{align*}
        \E[(\hat{\Phi}^{\Lag_i}- \Phi^{\Lag_i})(\hat{\Phi}^{\Lag_j} - \Phi^{\Lag_j})^\top] &= \frac{1}{T^2}\sum_{s,t=1}^T (\Sigma^{\Lag_i})^{-1} \E[\eta_i(t)\eta_j(s)\X^{\Lag_i}(t) (\X^{\Lag_j}(t))^\top] (\Sigma^{\Lag_j})^{-\top} \\ &+ O_p(T^{-2})
    \end{align*}
    Since $i < j$ it follows that $\eta_j(s)$ is independent of the variables $\eta_i(t)$ and $\X^{\Lag_j}(s)$ and $\X^{\Lag_i}(t)$, and since all those variables are jointly Gaussian it follows for all $s, t \in [1, T]$ that
    \begin{align*}
        \E[\eta_i(t)\eta_j(s)\X^{\Lag_i}(t) (\X^{\Lag_j}(t))^\top]= \E[\eta_j(s)] \E[\eta_i(t)\X^{\Lag_i}(t) (\X^{\Lag_j}(s))^\top] = 0, 
    \end{align*}
    since $\E[\eta_j(s)] = 0$. This finishes the proof. 
\end{proof}
\vspace*{-10pt}
\section{Asymptotic theory}
As mentioned already in the main paper, we will use the delta method \citep{Vaart_1998} to compute the asymptotic covariance of the link functions $\Fsep_{\bullet}(z)$.
\subsection{Asymptotic distributions of link function estimators}
Let $\G$ be a time series graph of order $p$ and let $G=(V, D)$ be the associated process graph. We fix a node $v\in V$ and a set $\Lag_v$ that satisfies condition (\ref{condition: time lags}) from the main paper. In view of the following arguments we introduce for every $u\in \pa(v)\cup \{v \} $ the ordered set
\begin{align*}
    \lag_u &= \left\{ k \mid (u,k) \in \Lag_v \right\} \subset [q],
\end{align*}
i.e. the set of time lags with which process $u$ is drving process $v$. 

For the computation of the asymptotic covariance of the estimator $\hat{\Fsep}_{\bullet, v}(z)$ with the delta method we need to compute  for every $u \in \pa(v)$ the Jacobian $\nabla^{\Lag_v} (\fsep_{u, v}(z))$ and the Jacobian $\nabla^{\Lag_v} (\intsep_{v}(z))$ , i.e. the derivative with respect to the parameters $\Phi^{\Lag_v}$. As a preparation we recall some simple but useful identities regarding multiplication and division of complex numbers. Suppose $z = x +iy \in \C\setminus\{0\}$, then its inverse is given as follows
\begin{align*}
    z^{-1} &= \frac{x - iy}{x^2 +y^2}. 
\end{align*}
Furthermore, if $z' = x' + iy'\in \C$, then 
\begin{align*}
    zz' &= (xx' -yy') + (xy'+ yx')i. 
\end{align*}
For some frequency $z \in S^1$ and an ordered set of time lags $\lag = \{l_1< \dots < l_m\}\subset[q]$, we denote as $\mathbf{z}^\lag$ the vectors
\begin{align*}
    \mathbf{z}^\lag &\coloneqq \begin{bmatrix}
        \re(\mathbf{z}^\lag) \\ \imag(\mathbf{z}^\lag) 
    \end{bmatrix}=\begin{bmatrix}
        \re(z^{l_1}) & \cdots & \re(z^{l_m}) \\
        \imag(z^{l_1}) & \cdots & \imag(z^{l_m})
    \end{bmatrix} \in \R^{2 \times |\lag|}
\end{align*}  
If $f$ is a function in the parameters $\Phi^{\Lag_v}$, then we denote the partial derivatives of $f$ with respect to $(\phi^{\Lag_w}_{u,v}(k))_{(u,k) \in \lag_u}$ as
\begin{align*}
    \nabla^{\lag_u} f & \coloneqq \left(\frac{\partial}{\partial \phi^{\Lag_v}_{u,v}(k)}f\right)_{k \in \lag_u}.
\end{align*}
\begin{lemma}\label{lemma: derivative transfer function}
    For $u \in \pa(v)$, the non-zero partial derivatives of the complex transfer function $\fsep_{u,v}(z)$ evaluated at $z \in S^1$ are as follows 
    \begin{align*}
        \nabla^{\lag_u}\fsep_{u,v}(z) &= a_v(z)A(z) \mathbf{z}^{\lag_u} & \nabla^{\lag_v} \fsep_{u,v}(z) &= a_v(z)A_u(z) \mathbf{z}^{\lag_v}, 
    \end{align*}
    where
    \begin{align*}
        a_v(z) &= |1- \varphi_{v,v}(z)|^{-2} \\
        A(z) &= \begin{bmatrix}
            \re(1- \varphi_{v,v}(z)) & -\imag(1- \varphi_{v,v}(z)) \\
            \imag(1- \varphi_{v,v}(z)) & \re(1- \varphi_{v,v}(z))
        \end{bmatrix} \\
        A_{u}(z) &= \begin{bmatrix}
            -\re(\varphi_{u,v}(z)) & \imag(\varphi_{u,v}(z)) \\
            \imag(\varphi_{u,v}(z)) & -\re(\varphi_{u,v}(z))
        \end{bmatrix} -2\begin{bmatrix}
             \re(\fsep_{u, v}(z)) & \re(\fsep_{u, v}(z)) \\
             \imag(\fsep_{u, v}(z))& \imag(\fsep_{u, v}(z))
         \end{bmatrix} 
    \end{align*}
    The non-zero partial derivatives of the internal function of $v$, i.e. $\intsep_v$ are as follows 
    \begin{align*}
        \nabla^{\lag_v} \intsep_v &= a_v(z)A_v(z) \mathbf{z}^{\lag_v},
    \end{align*}
    where
    \begin{align*}
        A_v(z) &=\left(\begin{bmatrix}
            1 & 0 \\ 0 & -1
        \end{bmatrix} - 2 \begin{bmatrix}
            \re(\intsep_v(z)) & \re(\intsep_v(z)) \\ - \imag(\intsep_v(z)) & - \imag(\intsep_v(z))
        \end{bmatrix}\right).
    \end{align*}
\end{lemma}
\begin{proof}
    Let $u\in \pa(v)$ and $\lag_u$ and $\lag_v$ be as defined above. In the following we will denote the real and imaginary parts of the complex valued link function as
    \begin{align*}
        \rfsep_{u,v}(z) &\coloneqq \mathrm{Re}(\fsep_{u,v}(z)) & \ifsep_{u,v}(z) &\coloneqq \mathrm{Im}(\fsep_{u,v}(z)) 
    \end{align*} 
    The rules for complex multiplication and division with respect to real and imaginary parts allow us to decompose the real and imaginary part of the complex transfer function as 
    \begin{align} \label{eq: real and complex transfer}
        \rfsep_{u,v}(z) &= \alpha_v(z) \tilde{\rfsep}_{u,v}(z) & \ifsep_{u,v}(z) &= \alpha_{v}(z) \tilde{\ifsep}_{u,v}(z),
    \end{align}
    where 
    \begin{align*}
        \tilde{\rfsep}_{u,v}(z) &= \mathrm{Re}(\varphi_{u,v}(z)) \mathrm{Re}(1- \varphi_{v,v}(z)) + \mathrm{Im}(\varphi_{u,v}(z))\mathrm{Im}(\varphi_{v,v}(z)) \\
        \tilde{\ifsep}_{u,v}(z) &= -\mathrm{Re}(\varphi_{u,v}(z))\mathrm{Im}(\varphi_{v,v}(z)) + \mathrm{Im}(\varphi_{u,v})\mathrm{Re}(1- \varphi_{v,v}(z)), 
    \end{align*}
    and
    \begin{align*}
        \alpha_v(z) &= |1- \varphi_{v,v}(z)|^{-2}.  
    \end{align*}

    The vector of partial derivatives with respect to the coefficients $(\phi^{\Lag_v}_{u,v}(k))_{k \in \lag_u}$ are linear combinations of the vectors $\re(\mathbf{z}^{\lag_v})$ and $\imag(\mathbf{z}^{\lag_v})$, i.e., 
    \begin{align*}
        \nabla^{\lag_u}\tilde{\rfsep}_{u,v}(z) &= \re(1 - \varphi_{v,v}(z) \re(\mathbf{z}^{\lag_u}) + \imag(\varphi_{v,v}(z)) \imag(\mathbf{z}^{\lag_v}) \\
        \nabla^{\lag_u}\tilde{\ifsep}_{u,v}(z) &= -\imag(\varphi_{v,v}(z)) \re(\mathbf{z}^{\lag_u})+ \re(1-\varphi_{v,v}(z)) \imag(\mathbf{z}^{\lag_u})
    \end{align*}
     From Equation (\ref{eq: real and complex transfer}) and the fact that $\alpha_v(z)$ is a function of the variables $(\phi^{\Lag_v}_{v,v}(k))_{k \in \lag_v}$ only, we obtain that
    \begin{align*}
        \nabla^{\lag_u} \fsep_{u,v}(z) &= \alpha_v(z)(1- \varphi_{v,v}(z)) \mathbf{z}^{\lag_u} 
    \end{align*}
    which shows the first part of the lemma. 
    
    We proceed with the computation of the partial derivatives with respect to the coefficients $(\Phi^{\Lag_v}_{v,v}(k))_{k \in \lag_v}$, which is also a linear combination of the vectors $\re(\mathbf{z}^{\lag_v})$ and $\imag(\mathbf{z}^{\lag_v})$, i.e.  
    \begin{align*}
        \nabla^{\lag_v}\tilde{\rfsep}_{u,v}(z) &= -\re(\varphi_{u,v}(z)) \re(\mathbf{z}^{\lag_v}) + \imag(\varphi_{u,v}(z)) \imag(\mathbf{z}^{\lag_v}) \\
        \nabla^{\lag_v}\tilde{\ifsep}_{u,v}(z) &= -\re(\varphi_{u,v}(z) \imag(\mathbf{z}^{\lag_v}) + \imag(\varphi_{u,v}(z)) \re(\mathbf{z}^{\lag_v})  \\
    \end{align*}
    Applying the product rule to Equation (\ref{eq: real and complex transfer}) yields  
    \begin{align}\label{eq: proof derivative transfer function}
        \nabla^{\lag_v} \fsep_{u,v}(z) &= \alpha_{v}(z)\begin{bmatrix}
            \nabla^{\lag_v}\tilde{\rfsep}_{u,v}(z) \\  \nabla^{\lag_v}\tilde{\ifsep}_{u,v}(z)
        \end{bmatrix}\left[\nabla^{\lag_v}\tilde{\fsep}_{u,v}(z) \right] + \begin{bmatrix}
             \tilde{\rfsep}_{u,v}(z) \\  \tilde{\ifsep}_{u,v}(z)
        \end{bmatrix} \left[\nabla^{\lag_v}\alpha_v(z)\right].
    \end{align}
    Furthermore, we compute that 
    \begin{align*}
        \nabla^{\lag_v} a_v(z) = -2a_v(z)^2 \begin{bmatrix}
            1 & 1 \\
            0 & 0 
        \end{bmatrix}
        \mathbf{z}^{\lag_v}.
    \end{align*}
    As a result we get for the second summand in Equation (\ref{eq: proof derivative transfer function}) the following expression:
    \begin{align*}
         \begin{bmatrix}
             \tilde{\rfsep}_{u,v}(z) \\  \tilde{\ifsep}_{u,v}(z)
        \end{bmatrix} \left[\nabla^{\lag_v}\alpha_v(z)\right] &= \begin{bmatrix}
             \tilde{\rfsep}_{u, v}(z) & \tilde{\rfsep}_{u, v}(z) \\
             \tilde{\ifsep}_{u, v}(z) & \tilde{\ifsep}_{u, v}(z)
         \end{bmatrix} \mathbf{z}^\lag
    \end{align*}
    So, we conclude with the following formula for the Jacoobian  
    \begin{align*}
        \nabla^{\lag_v} \fsep_{u, v} &= \alpha_v(z)(\begin{bmatrix}
            -\re(\varphi_{u,v}(z)) & \imag(\varphi_{u,v}(z)) \\
            \imag(\varphi_{u,v}(z)) & -\re(\varphi_{u,v}(z))
        \end{bmatrix} -2\begin{bmatrix}
             \rfsep_{u, v}(z) & \rfsep_{u, v}(z) \\
             \ifsep_{u, v}(z) & \ifsep_{u, v}(z)
         \end{bmatrix} )\mathbf{z}^\lag
    \end{align*}
    The derivative of the function $\intsep_v(z)$ is computed in a similar manner. 
\end{proof}
Let $z \in S^1$ be a frequency and $u_1, u_2\in \pa(v) \cup \{v\}$ two processes, then we define the $z$-dependent matrix  
\begin{align*}
    P^{\Lag_v}_{u_1,u_2}(z) \coloneqq \left[\mathbf{z}^{\lag_{u_1}} \right] P^{\Lag_v}_{u_1,u_2}  \left[ \mathbf{z}^{\lag_{u_2}} \right]^\top \in \R^{2 \times 2}. 
\end{align*}
The following Proposition is a more detailed version of Proposition \ref{prop: asymptotic normality link functions} in that it states the exact expressions for the asymptotic covariance of the link and internal function estimators.
\begin{proposition}\label{prop: asymptotic covariance link functions}
    Let $A: S^1 \to \R^{2\times2}$ and  $A_u: S^1 \to \R^{2\times2}$ for $u \in \pa(v) \cup \{v\}$ be the functions from Lemma \ref{lemma: derivative transfer function}, and for $u_1,u_2 \in \pa(v) \cup \{v\}$ the asymptotic covariance of the corresponding link function is given as
    \begin{align*}
            \acov^{\Lag_v}(u_1, u_2; z) &= a_v(z)^2 \begin{bmatrix}
                A(z) & A_{u_1}(z)
            \end{bmatrix}
            \begin{bmatrix}
                P^{\Lag_v}_{u_1,u_2}(z) & P^{\Lag_v}_{u_1,v}(z)  \\
                P^{\Lag_v}_{v,u_2}(z) & P^{\Lag_v}_{v,v}(z) 
            \end{bmatrix}
            \begin{bmatrix}
                A(z) & A_{u_2}(z)
            \end{bmatrix}^\top
    \end{align*}
    If $u \in \pa(v)$, then  
    \begin{align*}
        \acov^{\Lag_v}(u, v; z) & = a_v(z)^2 \begin{bmatrix}
                A(z) & A_{u}(z)
            \end{bmatrix}
            \begin{bmatrix}
                P^{\Lag_v}_{u,v}(z)  \\
                P^{\Lag_v}_{v,v}(z) 
            \end{bmatrix}
            \begin{bmatrix}
                A_{v}(z)
            \end{bmatrix}^\top
    \end{align*}
    Finally,  
    \begin{align*}
        \acov^{\Lag_v}(v,v;z) &= a_v(z)^2 \begin{bmatrix}
             A_{v}(z)
            \end{bmatrix}
            P^\Lag_{v,v}(z)
            \begin{bmatrix}
                A_{v}(z)
            \end{bmatrix}^\top
    \end{align*}
\end{proposition}
\begin{proof}
    In Lemma \ref{lemma: derivative transfer function} we computed the derivative of the function $\Fsep_{\bullet, v}(z)$ as a function of the parameters $\Phi^\Lag$. By the delta method, the asymptotic covariance of $\hat{\Fsep}_{\bullet, v}(z)$ is as follows 
    \begin{align*}
        \acov^{\Lag_v}(z) &= [\nabla^{\Lag_v} \Fsep_{\bullet, v}(z)] P^{\Lag_v} [\nabla^{\Lag_v}\Fsep_{\bullet, v}(z)]^\top \\
        &= [\acov^{\Lag_v}(u_1,u_2)]_{u_1, u_2 \in \pa(v) \cup \{v\}} \in \R^{(\pa(v)\cup \{v\}) \times (\pa(v)\cup \{v\})} \otimes \R^{2 \times 2}.
    \end{align*}
    Suppose $u_1,u_2 \in \pa(v)$, then the corresponding entry in the asymptotic block covariance matrix is
    \begin{align*}
        \acov^{\Lag_v}(u_1, u_2; z) &=[\nabla^{\Lag_v} \fsep_{u_1, v}(z)] P^{\Lag_v} [\nabla^{\Lag_v} \fsep_{u_1, v}(z)]^\top \\
        &= A(z)(\mathbf{z}_v^{\lag_{u_1}} P^{\Lag_v}_{u_1,u_2} [\mathbf{z}_v^{\lag_{u_2}}]^\top) A(z)^\top + A_{u_1}(z) (\mathbf{z}_v^{\lag_v} P^{\Lag_v}_{v,v} [\mathbf{z}_v^{\lag_v}]^\top) A_{u_2}(z)^\top \\
        &+ A(z)(\mathbf{z}_v^{\lag_{u_1}} P^{\Lag_v}_{u_1,v} [\mathbf{z}_v^{\lag_v}]^\top) A_{u_2}(z)^\top + A_{u_1}(z) (\mathbf{z}_v^{\lag_v} P^{\Lag_v}_{v,u_2} [\mathbf{z}_v^{\lag_{u_2}}]^\top) A(z)^\top \\
        &= \begin{bmatrix}
            A(z) & A_{u_1}(z)
        \end{bmatrix} \begin{bmatrix}
            P_{u_1,u_2}^{\Lag_w}(z) & P_{u_1,v}^{\Lag_w}(z) \\
            P_{v,u_2}^{\Lag_w}(z) & P_{v,v}^{\Lag_w}(z)
        \end{bmatrix}
        \begin{bmatrix}
            A(z)^\top \\
            A_{u_2}(z)^\top
        \end{bmatrix}
    \end{align*}
    Similarly, we compute for any $u \in \pa(v)$ that 
    \begin{align*}
        \acov^{\Lag_v}(u,v; z) = \begin{bmatrix}
            A(z) & A_u(z) 
        \end{bmatrix} 
        \begin{bmatrix}
            P^{\Lag_v}_{u, v} \\ P^{\Lag_v}_{v,v} 
        \end{bmatrix}
        \begin{bmatrix}
            A_v(z)^\top 
        \end{bmatrix} = \acov^{\Lag_v}(v, u; z)^\top,
    \end{align*}
    and finally
    \begin{align*}
        \acov^{\Lag_v}(v,v,z) &= \left[A_v(z)\right] \left[ P^{\Lag_v}_{v,v}(z) \right] \left[ A_v(z) \right] ^\top.
    \end{align*}
\end{proof}
It remains to prove that the link function estimator is asymptotically normal for generic parameter choices. First, let $v\in V$ be a process index and $\pa(v) =\{u_1, \dots, u_n \}$ its parent processes. We wish to write the asymptotic covariance as a matrix product like this
\begin{align*}
    \acov^{\Lag_v}(z) &= \alpha_v(z)^2\left[\mathbf{A}(z)\right] \mathbf{Z} P^{\Lag_v} \mathbf{Z}^\top \left[\mathbf{A}(z) \right]^\top.
\end{align*}
In order to see that the estimate of the link-functions and the internal function follows a joint normal distribution for generic parameter choices, we need to show that the asymptotic covariance matrix is invertible for generic parameter choices. To get the matrix product above, we need some notation for block structured matrices. 

Suppose we have disjoint sets $I_1, \cdots, I_m$ and $J_1, \dots, J_m$, and matrices $A_1, \dots, A_m$ such that $A_k \in \R^{I_k \times J_k}$, then we define the block diagonal matrix $\mathrm{diag}(A_i)_{i \in [m]} \in \R^{I \times J}$, where $I = \cup_{k=1}^m I_k$ and $J = \bigcup_{l=1}^m J_l$, at the entry $(i,j)$ to be $[A_k]_{i,j}$ if $(i,j) \in I_k \times J_k$ for some $k$, and zero otherwise. Let us consider an example. 
\begin{example}
    Suppose $I_1 = [(a,1), (a,2)]$ and $I_2 = [(b,1), (b,2), (b,3)]$ and $J_1 = [(c, 1), (c, 2)]$ and $J_2 = [(d, 1), (d,2)]$. Further let 
    \begin{align*}
        A_1 &= \begin{bmatrix}
            1 & 1 \\
            1 & 1
        \end{bmatrix} \in \R^{I_1 \times J_1} & 
        A_2 &= \begin{bmatrix}
            1 & 1 \\
            1 & 1 \\
            1 & 1
        \end{bmatrix} \in \R^{I_2 \times J_2}
    \end{align*}
    Accordingly, 
    \begin{align*}
        \mathrm{diag}(A_1, A_2) &= \begin{bmatrix}
            A_1 & \mathbf{0} \\
            \mathbf{0} & A_2
        \end{bmatrix}
        = \begin{bmatrix}
            1 & 1 & 0 & 0 \\
            1 & 1 & 0 & 0 \\
            0 & 0 & 1 & 1 \\
            0 & 0 & 1 & 1 \\
            0 & 0 & 1 & 1 
        \end{bmatrix}
    \end{align*}
\end{example}
Similarly, we define vertical stacking of matrices. Suppose $I_1, \dots, I_m$ are mutually disjoint finite sets and $J$ a finite set, and let $B_{1}, \dots B_m$ be matrices such that $B_k \in \R^{I_k \times J}$. Then we define $\mathrm{vst}(B_{k}) \in \R^{I \times J}$, with $I= \bigcup_{k=1}^m I_k$, to be the matrix that one obtains by vertically stacking the matrices $B_k$.

With this notation we define the block upper triangular matrices 
\begin{align*}
    \mathbf{A}(z) & \coloneqq\begin{bmatrix}
        \mathrm{diag}\left[A(z)\right]_{u \in \pa(v)} & \mathrm{vst}\left[ A_{u}(z) \right]_{u \in \pa(v)} \\
        \mathbf{0} & A_v(z)
    \end{bmatrix} & \mathbf{Z} & \coloneqq \begin{bmatrix}
        \mathrm{diag}\left[\mathbf{z}^{\lag_u}\right]_{u \in \pa(v)} & \mathrm{vst}\left[ \mathbf{z}^{\lag_v}\right]_{u \in \pa(v)} \\
        \mathbf{0} & \mathbf{z}^{\lag_v}
    \end{bmatrix}
\end{align*}
Since the matrix $P^{\Lag_v}$ is a the inverse of a covariance matrix it is positive definite. So the matrix product
\begin{align}\label{eq: frequency precision matrix invertible}
    \mathbf{Z} P^{\Lag_v} \mathbf{Z}^\top
\end{align}
is positive definite if the matrix $\mathbf{Z}$ has full rank, i.e. if it has rank $2(n+1)$. This is the case if $\Lag_v$ satisfies Assumption \ref{assumption: AL2}, and, additionally if $z$ is such that for every $u \in \pa(v) \cup \{ v \}$ there are $k,j \in \lag_u$ such that 
\begin{align*}
    0 &\neq z^k - z^j & 0 \neq z^k + z^j. 
\end{align*}
This requirement makes sure that there are two linear independent columns in $\mathbf{z}^{\lag_u} \in \R^{2 \times |\lag_u|}$. Note that for $|\lag_u| \geq 2$ the matrix $\mathbf{Z} P^{\Lag_v} \mathbf{Z}^\top$ is positive definite for all but finitely many $z\in S^1$. 
These conditions do not depend on the parameter $\Phi$ of the SVAR process. Instead they depend on the choice of $\Lag_v$.
Finally, the matrix $\mathbf{A}(z)$ is invertible if and only if its determinant is non-zero. And its determinant is a rational function in the parameters $\Phi^{\Lag_v}$. Specifically, 
\begin{align*}
    \det(\mathbf{A}(z)) &= \det(A(z))^{n} \det(A_{v}(z))
\end{align*}
By (\ref{assimption: stability}) it holds that $\det(A(z)) = |1- \varphi_{v,v}(z)|^2 \neq 0$. So, to show that the matrix $\mathbf{A}(z)$ is invertible, we need to show that $\det(A_v(z))$ is not the zero function (as a rational function of the indeterminates $\phi^{\Lag_v}_{v,v}(k)$). So let us compute the determinant explicitly
\begin{align*}
    \det(A_v(z)) &= \det\begin{bmatrix}
        1 - 2 \re(\intsep_v(z)) & - 2 \re(\intsep_v(z)) \\
        2 \imag(\intsep_v(z)) & -1 + 2\re(\intsep_v(z))
    \end{bmatrix} \\
    &= - 1 + 2\re(\intsep_v(z)) + 2\imag(\intsep_v(z))
\end{align*}
Since $\intsep_v = (1- \varphi_{v,v})^{-1}$ it follows that $\det(A_v(z))$ is non-zero if and only if the following expression is non-zero
\begin{align*}
    P(\phi_{v,v}^{\Lag_v}) &= - |1- \varphi_{v,v}(z)|^2 + 2(1- \re(\varphi_{v,v}(z))) + 2 \imag(\varphi_{v,v}(z)) \\
    &= -1 + \sum_{l \in \lag_v} \beta_l\phi^{\Lag_v}_{v,v}(l) + \sum_{k,j \in \lag_v} \gamma_{j,k} \phi^{\Lag_v}(k)\phi^{\Lag_v}_{v,v}(j),
\end{align*}
where $\beta_l= \imag(z^l)$ and $\gamma_{j,k} = \re(z^j)\re(z^k) + \imag(z^j)\imag(z^k)$. In particular, $P$ is a non-zero polynomial in $\phi^{\Lag_v}(k)$, so for generic choices of stable SVAR parameters it is non-zero. So, if $\Lag_v$ satisfies Assumption \ref{assumption: AL2}, then for almost all $z\in S^1$ and almost all SVAR parameter configurations $\Phi$, the OLS-based estimator $\hat{\Fsep_{\bullet, v}}(z)$ is asymptotically normal with block covariance matrix $\acov^{\Lag_v}(\bullet, \bullet; z)$. 

\subsection{Proof of Proposition \ref{prop: asymptotic covariance path function} and Proposition \ref{prop: asymptotic covariance weighted path function}}
Let $\G$ be a time series DAG with associated process graph $G= (V,D)$. For each process $v \in V$, we select finite sets of lagged parents $\Lag_v \supseteq \Lag_v^\G$ that satisfy condition (\ref{condition: time lags}).   
We will now prove Proposition \ref{prop: asymptotic covariance weighted path function}. The proof of Proposition \ref{prop: asymptotic covariance path function} is done in exactly the same way.  
\begin{proof}[Proof of Proposition \ref{prop: asymptotic covariance weighted path function}]
    
Suppose $\pi = x_1 \to \cdots \to x_n$ and $\rho = y_1 \to \cdots \to y_q$ are directed paths on $G$. In the following we use $P^\Lag$ to denote the matrix
$\mathrm{diag}[\omega_2\left(\Sigma^{\Lag_1}\right)^{-1} , \dots , \omega_m \left(\Sigma^{\Lag_m}\right)^{-1}]$, which can be considered as a block matrix, where each block is indexed by a pair of numbers in $[1,m]$, i.e. 
\begin{align*}
    P^{\Lag} = [P^\Lag_{i,j}]_{i,j \in [1,m]},
\end{align*} where 
\begin{align*}
    P^{\Lag}_{i,j}= \begin{cases}
        P^{\Lag_i} =\omega_i (\Sigma^{\Lag_i})^{-1} & \text{if } i = j \\
        \mathbf{0} & \text{if } i \neq j
    \end{cases}
\end{align*} Furthermore, we consider submatrices of $P^\Lag$ that correspond to pairs of paths $\pi$ and $\rho$, and denote them as follows  
\begin{align*}
    P^{\Lag(\pi), \Lag(\rho)} = [P^\Lag_{x_i, y_j}]_{i\in [1,m], j \in[1, q]}. 
\end{align*}
Note that $P^{\Lag(\pi), \Lag(\rho)}$ is zero if $\pi$ and $\rho$ do not share a single vertex. On the other hand, if $\pi = \rho$, then $P^{\Lag(\pi), \Lag(\rho)}$ is block diagonal. By the delta method, the asymptotic covariance of the estimators $\hat{\csep}^{(\pi)}$ and $\hat{\csep}^{(\rho)}$ is
\begin{equation} \label{eq: proof asymptotic covariance spectral contribution}
    \begin{split}
        \acov^{\Lag}(\pi^\internal, \rho^\internal) &= \nabla^{\Lag}\csep^{(\pi)} P^\Lag (\nabla^\Lag \csep^{(\rho)})^\top \\
        &= \nabla^{\Lag(\pi)} \csep^{(\pi)} P^{\Lag(\pi), \Lag(\rho)} (\nabla^{\Lag(\rho)}\csep^{(\rho)})^\top \\
        &= \sum_{i=1}^{m}\sum_{j=1}^{q} (\nabla^{\Lag_{x_i}} \csep^{(\pi)}) P^{\Lag}_{x_i, y_j} (\nabla^{\Lag_{y_j}}\csep^{(\rho)})^\top \\
        &= \sum_{x \in V(\pi \cap \rho)} (\nabla^{\Lag_{x}} \csep^{(\pi)}) P^{\Lag}_{x, x} (\nabla^{\Lag_{x}}\csep^{(\rho)})^\top
    \end{split}
\end{equation}
where $\nabla^{\Lag(\pi)}$ resp. $\nabla^{\Lag(\rho)}$ denotes taking derivative with respect to the parameters $\Phi^{\Lag_{x_i}}$ resp. $\Phi^{\Lag_{y_j}}$, and $V(\pi  \cap \rho) = \{x_i\}_{i \in [m]} \cap \{y_j \}_{j \in [q]}$. This holds as $\csep^{(\pi)}$ resp. $\csep^{(\rho)}$ are functions only in the parameters $\Phi^{\Lag_{x_i}}$ resp. $\Phi^{\Lag_{y_j}}$, and because of the block diagonal structure of $P^{\Lag}$. 

We will now calculate $\nabla^{\Lag_x} \csep^{(\pi)}$ for every $x$ visited by $\pi$. First, we observe that for two distinct functions $F: \R^{k} \to \R^{m \times m}$ and $G: \R^{l} \to \R^{m \times m}$ that the derivative of the point-wise product is $\nabla(FG) = (\nabla^{\mathbf{x}}(F) G , F \nabla^{\mathbf{y}} (G))$. In addition, if $F(\mathbf{x}) G(\mathbf{y}) = G(\mathbf{y})F(\mathbf{x})$, then its derivative is $\nabla(FG) = (G \nabla^{\mathbf{x}}(F), F \nabla^{\mathbf{y}} (G))$. 

Suppose $z_{k} = x_k + iy_k \in \C$ for $k \in [1,m]$ and $x + iy = \prod_{k=1}^m z_k$ is the product of the $z_k$, i.e. $x$ is the real part of the product over the $z_k$'s and $y$ the imaginary part of the product, then 
\begin{align*}
    \begin{bmatrix}
        x \\ y
    \end{bmatrix} &= \left(\prod_{k=1}^m Z_k\right) \begin{bmatrix}
        1 \\ 0
    \end{bmatrix} & Z_k &\coloneqq \begin{bmatrix}
        x_k & -y_k \\ y_k & x_k
    \end{bmatrix}
\end{align*}
Of course, all the matrices on the right hand side commute. 
Recall that the function $\csep^{(\pi)} = \intsep_{x_1}\prod_{i=1}^{m-1} \fsep_{x_i, x_{i+1}}$, where each term in the product is considered as a complex number. We wish to compute the derivative of its real and imaginary part. For that purpose we use the above observation and define for a link $u \to v$ on the process graph the matrices
\begin{align*}
    F_u &= \begin{bmatrix}
        \re(\intsep_{u}) & - \imag(\intsep_{u}) \\   \imag(\intsep_{u}) & \re(\intsep_{u}) 
    \end{bmatrix} &
    H_{u, v} = \begin{bmatrix}
        \re(\fsep_{u, v}) & - \imag(\fsep_{u, v}) \\ \imag(\fsep_{u, v}) & \re(\fsep_{u, v})
    \end{bmatrix}
\end{align*}
Let now $i \in [m-1]$, then we use the above observations to compute  
\begin{align} \label{eq: derivative  path function}
    \nabla^{\Lag_{x_i}} \csep^{(\pi)} &= \begin{cases}
        (\prod_{j =1}^{m}H_{x_j, x_{j+1}}) \nabla^{\Lag_{x1}} \intsep_{x_1}  & \text{if } i=1  \\
        (F_{x_1} \prod_{j \neq i-1}^m H_{x_j, x_{j+1}})  \nabla^{\Lag_{x_i}} \fsep_{x_{i-1}, x_{i}} & \text{if } i > 1
    \end{cases},
\end{align}
where we used the commutativity of the matrices and the fact the parameter $\Phi^{\Lag_{x_i}}$ only affects the factor $H_{x_{i-1}, x_i}$ if $ i >1$ and only $F$ if $i=1$. Plugging the expression from Equation \ref{eq: derivative  path function} into the last line of Equation \ref{eq: proof asymptotic covariance spectral contribution} we get 
\begin{align*}
    \acov^{\Lag}(\pi^\internal, \rho^\internal) &=\sum_{(i,j): x_i = y_j } \csep^{(\pi \setminus x_i)} \left[\nabla^{\Lag_{x_i}} \fsep_{x_{i-1}, x_i} \right] P^{\Lag}_{x_i, y_j} \left[\nabla^{\Lag_{y_j}} \fsep_{y_{j-1}, y_j}\right]^{\top}(\csep^{(\rho\setminus y_{j}) })^\top \\
    &= \sum_{(i,j): x_i = y_j } \csep^{(\pi \setminus x_{i})} \left[\acov^{\Lag_{x_i}}(x_{i-1}, y_{j-1})\right] (\csep^{(\rho\setminus y_{j}) })^\top, 
\end{align*}
where $x_0 = x_1$ and $y_0 = y_1$ and
\begin{align*}
    \fsep_{x_0, x_1} &\coloneqq \intsep_{x_1} & 
    \csep^{(\pi \setminus x_i )} &\coloneqq \begin{cases}
        \prod_{j =1}^{m}H_{x_j, x_{j+1}} & \text{if } i=1 \\
        F_{x_1} \prod_{j \neq i-1}^m H_{x_j, x_{j+1}} & \text{if } i > 1
    \end{cases}
\end{align*}
The notations regarding the path $\rho$ are analogous.
\end{proof}
\subsection{Proof Proposition \ref{prop: asymptotic covariance spectral contribution}}
\begin{lemma}
    Suppose $v\in V$ is a process with not necessarily distinct ancestors $u, u'\in \anc(v)$. Furthermore, let $\pi = \pi_1 + \pi_2$ and $\rho = \rho_1 + \rho_2$ concatenations of paths on the process graph, where $\pi_1 \in \paths(u,v)$ and $\rho_1 \in \paths(u',v)$, and $\pi_2, \rho_2 \in \Pi \subset \paths(v,w)$, then it holds that
    \begin{align}
        \acov^{\Lag}(\pi^\internal, \rho^\internal) &= \fsep^{(\pi_2)}\acov^\Lag(\pi_1^\internal, \rho_1^\internal)\left( \fsep^{(\rho_2)}\right)^\ast + \csep^{(\pi_1)}\acov^\Lag(\pi_2, \rho_2)\left(\csep^{(\rho_1)}\right)^\ast.
    \end{align}
\end{lemma}
\begin{proof}
    This follows by combining Proposition \ref{prop: asymptotic covariance path function} and Proposition \ref{prop: asymptotic covariance weighted path function}. 
\end{proof}
\begin{proof}[Proposition \ref{prop: asymptotic covariance spectral contribution}]
    For $u, u' \in \anc(v)$ we compute 
    \begin{align*}
        \acov^\Lag(\hat{\csep}^{\Pi_u}, \hat{\csep}^{\Pi_{u'}}) &= \sum_{\pi \in \Pi_u} \sum_{\rho \in \Pi_{u'}}\acov^\Lag(\pi^\internal, \rho^\internal) \\
        &= \sum_{\pi_1 \in \paths(u,v)} \sum_{\pi_2 \in \Pi} \sum_{\rho_1 \in \paths(u',v)} \sum_{\rho_2 \in \Pi} \acov^\Lag((\pi_1 + \pi_2)^\internal, (\rho_1 + \rho_2)^\internal) \\ 
        &= \sum_{\pi_1 \in \paths(u,v)}\sum_{\rho_1 \in \paths(u',v)} \fsep^{\Pi}\acov^\Lag(\pi_1^\internal, \rho_1^\internal)\left( \fsep^{(\Pi)}\right)^\ast + \csep^{(\pi_1)}\acov^\Lag(\hat{\fsep}^{\Pi}, \hat{\fsep}^{\Pi},)\left(\csep^{(\rho_1)}\right)^\ast \\
        &= \fsep^\Pi [\acov^\Lag(\hat{\csep}^{\paths(u,v)}, \hat{\csep}^{\paths(u',v)})] (\fsep^\Pi)^\ast +\csep^{\paths(u,v)} [\acov(\hat{\fsep}^\Pi)] (\csep^{\paths(u',v)})^\ast
    \end{align*}
\end{proof}
\subsection{Asymptotic efficiency}
In this section, we apply the ideas from \cite{Guo2020EfficientLS} to prove Theorem \ref{thm: asymptotic efficiency} in the main paper. Suppose $\G=(V \times \Z, \D)$ is a time series DAG of order $p$ and let $q \geq p$. We assume that the vertices $V = \{v_1, \dots, v_m \}$ are topologically ordered with respect to the contemporaneous graph $\G_0$. Let $z \in S^1$ be a frequency and $\tau$ a vector of sums of (weighted) path functions, i.e. $[\fsep^{\Pi_u}]_{u \in U}$ or $[\csep^{\Pi_u}]_{u \in U}$, where $U\subset V$ is a set of processes and $\Pi_u\subset \paths(u,v)$ a set of paths such that each $\Pi_u$ corresponds to a controlled causal effect \cite{reiter2023formalising} of $u$ on $v$. This requirement on $\Pi_u$ guarantees that $\tau(z)$ is a rational function in the SVAR parameters. 

Finally, we pick an estimator $\hat{\tau}(z) \in \mathcal{T}^{(q)}(z)$ and precompose it with the diffeomorphism from Proposition \ref{prop: diffeomorphism} to get a function $\tilde\tau(z)$ in the variables $\Phi^{\tilde\Lag}, \Omega$, where $\tilde\Lag = \{ \tilde\Lag_v \}_{v \in V}$ is the collection of time lagged relations defined as in (\ref{eq: time lags coarse}).

The proof of Theorem \ref{thm: asymptotic efficiency} will be based on the delta method, which involves partial derivatives. Suppose $\Lag'$ is a collection of time lagged relations with corresponding (augmented) parameter vector $\Phi^{\Lag'}$. If $\tau'(z)$ is a (rational) function in $\Phi^{\Lag'}$ and the variance vector $\Omega$, then we denote the partial derivatives as follows 
\begin{align*}
    \nabla^{\Lag'}\tau'(z) &= \left[ \nabla^{\Lag'_v} \tau'(z) \right]_{v \in V}, & \nabla^{\Lag_v'} \tau'(z) &= \left[ \frac{\partial}{\partial \phi^{\Lag'_v}_{u,v}(k)} \tau'(z) \right]_{(u,k) \in \Lag'_v} \\
    \nabla^\Omega \tau'(z) &= \left[ \frac{\partial}{\partial \omega_v} \tau'(z) \right]_{v \in V}
\end{align*}
With these notations at hand we proceed with the proof. 
\begin{lemma}[Corollary 22 in \cite{Guo2020EfficientLS}]\label{lemma: jacobaian estimator}
    It holds that 
    \begin{align*}
        \nabla^{\Lag^\G}\tilde\tau(z) &= \nabla^{\Lag^\G}\hat{\tau}^\G(z) & \nabla^\Omega \tilde\tau(z) &= \mathbf{0}. 
    \end{align*}
\end{lemma}
\begin{lemma}[\cite{Guo2020EfficientLS} Lemma 27]\label{lemma: lower bound}
    Suppose $M \in \R^{I \times I}$ is a positive definite matrix and $\alpha, \beta \subset I$ define a partition of $I$, i.e. $ \alpha \cup \beta = I$ with $\alpha \cap \beta = \emptyset$. Then it holds for every $ x \in \R^I $ that 
    \begin{align*}
        x^\top M x \geq x_\alpha^\top M_{\alpha\cdot \beta } x_\alpha.
    \end{align*}
\end{lemma}
\begin{proof}[Proof of Theorem \ref{thm: asymptotic efficiency}]
Lets fix a vector $a\in \R^{2}\otimes \R^{U}$ and set $\tau_a(z) = a^\top \tau(z)$. Accordingly, we write $\tilde\tau_{a} (z) = a^\top \tilde\tau(z)$, and $\hat{\tau}^\Lag_a(z) = a^\top \hat{\tau}^\Lag(z)$. By combining the delta method with Lemma \ref{lemma: jacobaian estimator} we get that the asymptotic variance of the scalar valued estimator $ \tilde\tau_a(z)$ is  
\begin{equation}\label{eq: proof asymptotic variance decomposition}
    \begin{split}
        \mathrm{avar}( \tilde\tau_a(z)) &= \left[\nabla^{\tilde\Lag} \tilde\tau_a(z) \right] P^{\tilde\Lag} \left[\nabla^{\tilde\Lag} \tilde\tau_a(z) \right]^\top \\
    &=  \sum_{v \in V} \left[\nabla^{\tilde\Lag_v} \tilde\tau_a(z)\right] P^{\tilde\Lag_v} \left[\nabla^{\tilde\Lag_v} \tilde\tau_a(z)\right]^\top 
    \end{split}
\end{equation} 
Let us partition for every $v \in V$ the set of time lagged relations $\tilde\Lag_v$ using the sets $\alpha_v = \Lag^\G_v $ and $\beta_v = \tilde\Lag_v \setminus \Lag^\G_v$, so that 
\begin{equation}\label{eq: proof inequality}
    \begin{split}
        \left[\nabla^{\tilde\Lag_v} \tilde\tau_a(z)\right] P^{\tilde\Lag_v} \left[\nabla^{\tilde\Lag_v} \tilde\tau_a(z)\right]^\top &\geq 
        \left[\nabla^{\Lag_v} \tilde\tau_a(z)\right] P^{\tilde\Lag_v} \left[ P^{\tilde\Lag_v} \right]_{\alpha_v\cdot \beta_v} \left[\nabla^{\Lag_v} \tilde\tau_a(z)\right] \\
        &= \left[\nabla^{\Lag_v} \tilde\tau_a(z)\right] P^{\tilde\Lag_v} \left[ P^{\Lag_v^\G} \right] \left[\nabla^{\Lag_v} \tilde\tau_a(z)\right],
    \end{split}
\end{equation}
where the first inequality follows from Lemma \ref{lemma: lower bound} and the second equality follows from Corollary \ref{cor: schur inverison}. 
So combining equation (\ref{eq: proof asymptotic variance decomposition}) with the inequalities (\ref{eq: proof inequality}) yields 
\begin{align*}
    \mathrm{avar}( \tilde\tau_a(z)) &\geq \sum_{v \in V} \left[\nabla^{\Lag_v} \tilde\tau_a(z)\right] P^{\tilde\Lag_v} \left[ P^{\Lag_v^\G} \right] \left[\nabla^{\Lag_v} \tilde\tau_a(z)\right] \\
    &= \mathrm{avar}(\hat{\tau}_a^\G(z)). 
\end{align*}
Finally, this allows us to conclude with the inequality 
\begin{align*}
    a^\top \left( \acov(\tilde\tau(z)) - \acov(\hat{\tau}^\G(z)) \right) a =\mathrm{avar}( \tilde\tau_a(z))- \mathrm{avar}(\hat{\tau}_a^\G(z)) \geq 0
\end{align*}
which shows the claim (\ref{eq: absolute asymptotic efficiency}) of Theorem \ref{thm: asymptotic efficiency}, and similar arguments can be used to show statement (\ref{eq: relative asymptotic efficiency}) of Theorem \ref{thm: asymptotic efficiency}.
\end{proof}

\vspace*{-10pt}
\section{Example and Application}\label{sec: supplement numerical example}
\subsection{Numerical Example}
In this Section, we provide more details on the example discussed in Section \ref{sec: asymptotic distribution}. In Figure \ref{fig: numerical example tsg} we show the underlying time series graph of the example. The true lagged parent set and the lagged parents set encoding our knowledge about $\G$ given that we know the process graph $G$ are as follows
\begin{align*}
    \Lag^\G_{u_1}& =  \{(u_1, 1) \} & \Lag_{u_1} &= \{u_1\} \times [1, 3] \\
    \Lag^\G_{u_2}&=  \{(u_1, 1) \} & \Lag_{u_2} &= \{u_1\} \times [1, 3] \\
    \Lag^\G_{v} &= \{(v, 1), (v, 3), (u_1, 2), (u_2, 1)\} & \Lag_{v} &= \{v\} \times [1, 3] \cup \{u_1\} \times [0, 2] \cup \{u_2\} \times [0, 3] \\
    \Lag^\G_{m} &= \{(m, 1), (v, 3), (v, 1), (v, 2)\} & \Lag_{m}& = \{m\} \times [1, 3] \cup \{v\} \times [0, 3] \\
    \Lag^\G_{w} &= \{(w, 2), (v, 3), (v, 2), (m, 3)\} & \Lag_{w}& = \{w\} \times [1, 3] \cup \{v\} \times [0, 3] \cup \{m\} \times [0, 3] \\
\end{align*}
To generate the time series we used a white noise process with $\Omega = \mathbb{I}$ and we used the following SVAR coefficients 
\begin{align*}
    \Phi_{\bullet, u_1} &= \begin{bmatrix}
        \phi_{u_1, u_1}(1) = 0.5
    \end{bmatrix} \\
    \Phi_{\bullet, u_2} &= \begin{bmatrix}
        \phi_{u_2, u_2}(1) = 0.5
    \end{bmatrix} \\
    \Phi_{\bullet, v} &= \begin{bmatrix}
        \phi_{v, v}(1) = 0.3 & \phi_{v, v}(3) = -0.5 & \phi_{u_1, v}(2) = -0.25 & \phi_{u_2, v}(1) =0.5
    \end{bmatrix} \\
    \Phi_{\bullet, m} &= \begin{bmatrix}
        \phi_{m, m}(1) = 0.5 & \phi_{v, m}(1) = 0.5 & \phi_{v, m}(2) = 0.6 
    \end{bmatrix} \\
    \Phi_{\bullet, w} &= \begin{bmatrix}
        \phi_{w, w}(2) = 0.5 & \phi_{v, w}(2) = 0.4 & \phi_{m, w}(3) = -0.3 
    \end{bmatrix}
\end{align*}

\begin{figure}
    \centering
    \includegraphics[width=0.8\textwidth]{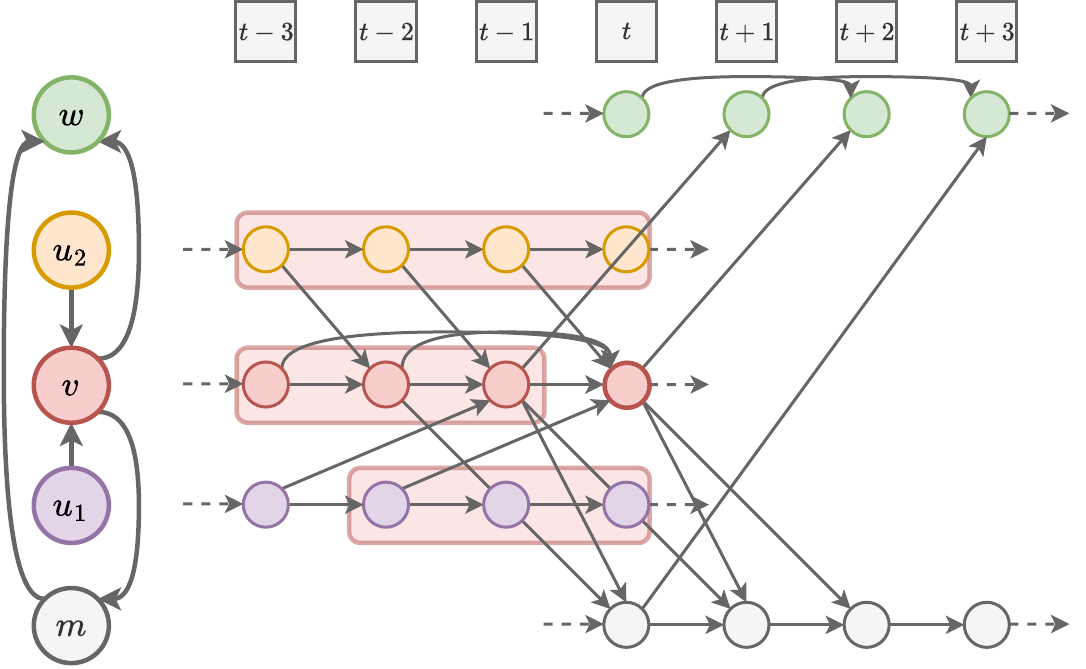}
    \caption{This shows time series underlying the process graph in Figure \ref{fig:process graph}. The red boxes represent the superset $\Lag_v \Lag_v^\G$}
    \label{fig: numerical example tsg}
\end{figure}
In the main paper we sketched the computation of the spectral contribution estimator. We will here provide the expressions for the asymptotic covariance matrices we omitted in the main part. First, we spell out the computation of the asymptotic covariance of the estimator $\hat{\fsep}^{\Pi} = \hat{\fsep}{v,w} + \hat{\fsep}_{v,m} \hat{\fsep}_{m,w}$. This amounts to 
\begin{align*}
    \acov^{\Lag}(\hat{\fsep}^{\Pi}, \hat{\fsep}^{\Pi}) &= \acov^\Lag(\pi_{v,w}, \pi_{v,w}) + \acov^\Lag(\pi, \pi) + \acov^\Lag(\pi_{v,w}, \pi) + \acov^\Lag(\pi,\pi_{v,w}) \\
    &= \acov^\Lag_w(v,v)  + \fsep_{v,m}\acov^\Lag_w(m,m)(\fsep_{v,w})^\ast + \fsep_{m,w}\acov^\Lag_m(v,v)(\fsep_{m,w})^\ast \\
    & + \acov^\Lag_w(v,m)(\fsep_{v,m})^\ast + \fsep_{v,m}\acov^\Lag_w(m,v)
\end{align*}
In the main paper we wanted to compute the asymptotic block covariance of the multivariate estimator of weighted path functions $(\hat{\csep}^{(\pi_{u_1, v})}, \hat{\csep}^{(\pi_{u_2, v})}, \hat{\csep}^{(\epsilon_{v})})$. We recall the expressions for the block entries on the diagonal from the main paper  
\begin{equation} \label{eq: example diagona acov weigthed path functions}
    \begin{split}
    \acov^\Lag(\pi_{u_1, v}^\internal, \pi_{u_1,v}^\internal) &= \intsep_{u_1}\acov^{\Lag_v}(u_1, u_1)\intsep_{u_1}^\ast \\
    \acov^\Lag(\pi_{u_2, v}^\internal, \pi_{u_2,v}^\internal) &= \intsep_{u_2}\acov^{\Lag_v}(u_2, u_2)\intsep_{u_2}^\ast \\
    \acov^\Lag(\epsilon_{v}^\internal, \epsilon_{v}^\internal) &= \acov^{\Lag_v}(v, v)
    \end{split}
\end{equation}
The off diagonal entries are then computed as follows 
\begin{equation} \label{eq: example off diagonal acov weighted path functions}
    \begin{split}
    \acov^\Lag(\pi_{u_1, v}^\internal, \pi_{u_2,v}^\internal) &= \intsep_{u_1}\acov^{\Lag_v}(u_1, u_2)\intsep_{u_2}^\ast \\
    \acov^\Lag(\pi_{u_1, v}^\internal, \epsilon_{v}^\internal) &= \intsep_{u_1}\acov^{\Lag_v}(u_1, v) \\
    \acov^\Lag(\pi_{u_2, v}^\internal, \epsilon_{v}^\internal) &= \intsep_{u_2}\acov^{\Lag_v}(u_2, v)
    \end{split}
\end{equation}
Finally, the expressions (\ref{eq: example diagona acov weigthed path functions})-(\ref{eq: example off diagonal acov weighted path functions}) together with Proposition \ref{prop: asymptotic covariance spectral contribution} let us compute the asymptotic covariance of the estimator $\hat{\csep}^{\Pi_{\anc(v)}}= (\hat{\csep}^{\Pi_{u_1}}, \hat{\csep}^{\Pi_{u_2}}, \hat{\csep}^{\Pi_{v}})$. The block diagonal entries of the asymptotic block covariance of the multi variate estimator are given by the following expressions
\begin{align*}
    \acov^{\Lag}(\hat{\csep}^{\Pi_{u_1}}, \hat{\csep}^{\Pi_{u_1}}) &= (\fsep^{\Pi}\intsep_{u_1})\acov^{\Lag_v}(u_1, u_1)(\intsep_{u_1}\fsep^{\Pi})^\ast +(\csep^{(\pi_{u_1,v})}) \acov^\Lag(\hat{\fsep^{\Pi}}, \hat{\fsep^{\Pi}}) (\csep^{(\pi_{u_1,v})})^\ast \\
    \acov^{\Lag}(\hat{\csep}^{\Pi_{u_2}}, \hat{\csep}^{\Pi_{u_2}}) &= (\fsep^{\Pi}\intsep_{u_2})\acov^{\Lag_v}(u_2, u_2)(\intsep_{u_2}\fsep^{\Pi})^\ast +(\csep^{(\pi_{u_2,v})}) \acov^\Lag(\hat{\fsep^{\Pi}}, \hat{\fsep^{\Pi}}) (\csep^{(\pi_{u_2,v})})^\ast \\
    \acov^{\Lag}(\hat{\csep}^{\Pi_{v}}, \hat{\csep}^{\Pi_{v}}) &= (\fsep^{\Pi})\acov^{\Lag_v}(v, v)(\fsep^{\Pi})^\ast +(\intsep_v) \acov^\Lag(\hat{\fsep^{\Pi}}, \hat{\fsep^{\Pi}}) (\intsep_v)^\ast
\end{align*}
The blocks on the off diagonal are the following
\begin{align*}
    \acov^{\Lag}(\hat{\csep}^{\Pi_{u_1}}, \hat{\csep}^{\Pi_{u_2}}) &= (\fsep^{\Pi}\intsep_{u_1})\acov^{\Lag_v}(u_1, u_2)(\intsep_{u_2}\fsep^{\Pi})^\ast +(\csep^{(\pi_{u_1,v})}) \acov^\Lag(\hat{\fsep^{\Pi}}, \hat{\fsep^{\Pi}}) (\csep^{(\pi_{u_2,v})})^\ast \\
    \acov^{\Lag}(\hat{\csep}^{\Pi_{u_1}}, \hat{\csep}^{\Pi_{v}}) &= (\fsep^{\Pi}\intsep_{u_1})\acov^{\Lag_v}(u_1, v)(\fsep^{\Pi})^\ast +(\csep^{(\pi_{u_1,v})}) \acov^\Lag(\hat{\fsep^{\Pi}}, \hat{\fsep^{\Pi}}) (\intsep_v)^\ast \\
    \acov^{\Lag}(\hat{\csep}^{\Pi_{u_2}}, \hat{\csep}^{\Pi_{v}}) &= (\fsep^{\Pi}\intsep_{u_2})\acov^{\Lag_v}(u_2, v)(\fsep^{\Pi})^\ast +(\csep^{(\pi_{u_2,v})}) \acov^\Lag(\hat{\fsep^{\Pi}}, \hat{\fsep^{\Pi}}) (\intsep_v)^\ast
\end{align*}

\subsection{Solar impact on the NAO}

\begin{figure}
    \centering
    \includegraphics[width=0.8\textwidth]{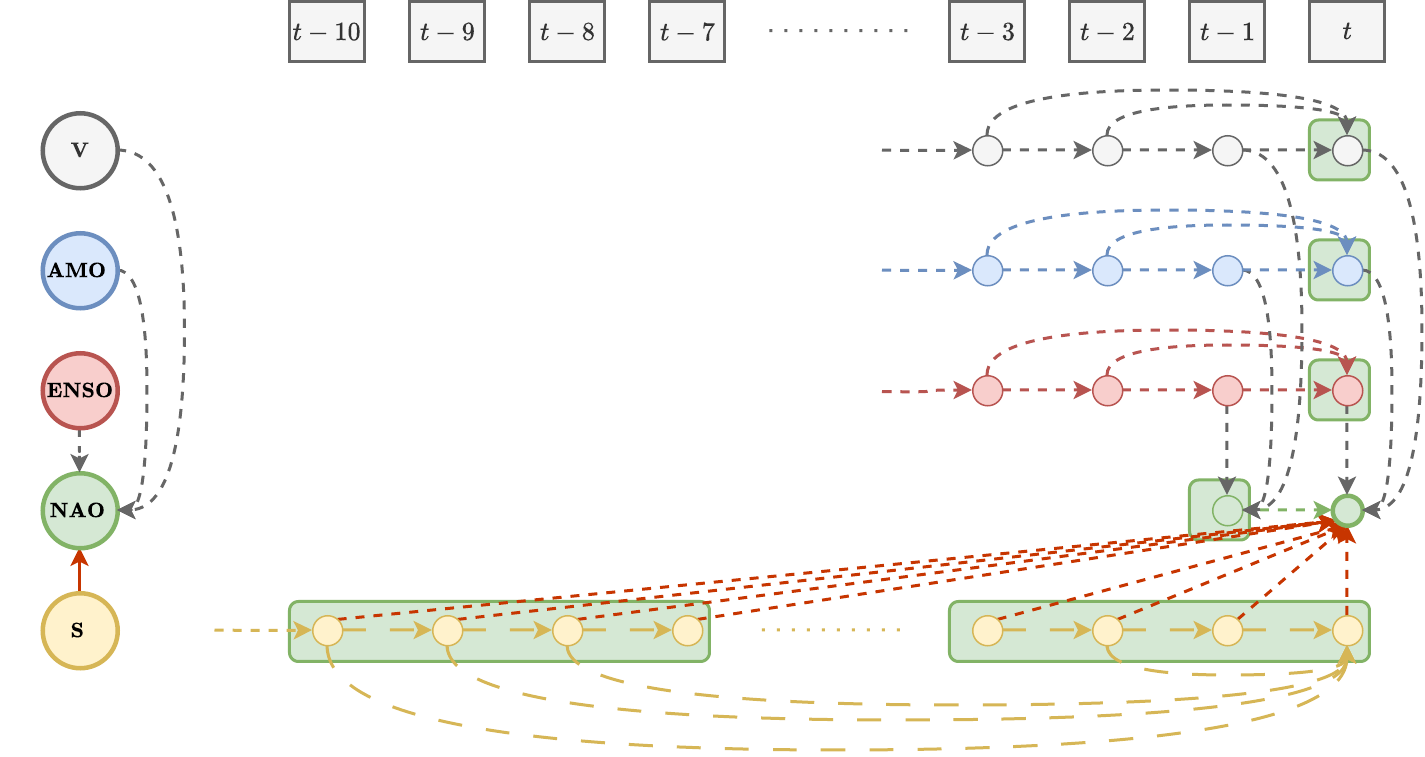}
    \caption{This is an excerpt of the time series graph $\G$ we used for the computation of the frequency domain causal effects. Since we do not know the true time series graph all the edges on $\G$}
    \label{fig: tsg solar nao}
\end{figure}

The parent sets encoding our assumptions about the time series graph are the following (see also Figure \ref{fig: tsg solar nao})
\begin{align*}
    \Lag_{\NAO} &= \{(\NAO, 1), (\NAO, 2), (\ENSO, 0), (\AMO, 0), (\AOD, 0) \} \cup \{(\SNS, k): k \in [0,3] \cup [7,10] \} \\
    \Lag_{\SNS} &= \{ (\SNS, k): k \in [1, 2] \cup [8,10] \} \\
    \Lag_{\ENSO} &= \{ (\ENSO, k): k \in [1, 3] \} \\
    \Lag_{\AMO} &= \{ (\AMO, k): k \in [1, 3] \} \\
    \Lag_{\AOD} &= \{ (\AOD, k): k \in [1, 3] \} \\
\end{align*}
Accordingly, the effect function of the link $\SNS \to \NAO$ is parameterised as follows
\begin{align*}
    \fsep_{\SNS \to \NAO}(z) &= (1- \Phi_{\NAO}(1)z - \Phi_{\NAO}(2)z^2)^{-1}(\phi_{\SNS \to \NAO}(0) + \phi_{\SNS \to \NAO}(1)z + \Phi_{\SNS \to \NAO}(2)z^2 + \\
    &+\Phi_{\SNS \to \NAO}(3)z^3 + \Phi_{\SNS \to \NAO}(7)z^7 + \Phi_{\SNS \to \NAO}(8)z^8 + \Phi_{\SNS \to \NAO}(9)z^9 + \Phi_{\SNS \to \NAO}(10)z^{10}),
\end{align*}
and the internal function of $\SNS$ is thus parameterised by the rational function 
\begin{align*}
    \intsep_{\SNS}(z) &= \frac{1}{1 - \Phi_{\SNS}(1)z - \Phi_{\SNS}(2)z^2 - \Phi_{\SNS}(1)z^8- \Phi_{\SNS}(8)z^8 - \Phi_{\SNS}(9)z^9 - \Phi_{\SNS}(10)z^{10}}
\end{align*}

\bibliographystyle{plain}
\bibliography{paper-ref}

\end{document}